\documentclass[reqno]{amsart}
\usepackage{comment,amssymb,latexsym,upref,enumerate,fouridx}
\usepackage{mathrsfs,color}
\usepackage{centernot}
\usepackage{xspace}

\usepackage[colorlinks,linkcolor=blue,citecolor=blue]{hyperref}
\allowdisplaybreaks
\numberwithin{equation}{section}

\theoremstyle{plain}
\newtheorem{thm}{Theorem}[section]

\newtheorem{prop}[thm]{Proposition}

\theoremstyle{definition}

\theoremstyle{remark}
\newtheorem{rem}[thm]{Remark}

\newcommand{\rmd}{\mathrm{d}}
\newcommand{\Ric}{\mathrm{Ric}}
\newcommand{\R}{\mathrm{R}}
\newcommand{\LKone}{A}
\newcommand{\LKtwo}{M}

\input Tdef.def

\begin{document}

\title[Stability of Asymptotic Behavior Within Polarised $\Tbb^2$-Symmetric Vacuum Solutions with Cosmological Constant]{Stability of Asymptotic Behavior Within Polarised $\Tbb^2$-Symmetric Vacuum Solutions with Cosmological Constant}

\author[E. Ames]{Ellery Ames}
\address{Dept. of Mathematics \\ 
Humboldt State University \\
1 Harpst St.
Arcata, CA 95521\\USA }
\email{Ellery.Ames@humboldt.edu}

\author[F. Beyer]{Florian Beyer}
\address{Dept. of Mathematics and Statistics\\
730 Cumberland St\\
University of Otago, Dunedin 9016\\ New Zealand}
\email{fbeyer@maths.otago.ac.nz }

\author[J. Isenberg]{James Isenberg}
\thanks{J. Isenberg is partially supported by NSF grant PHY-1707427.}
\address{Dept. of Mathematics and Institute for Fundamental Science \\
University of Oregon \\
Eugene, OR 97403 USA}
\email{isenberg@uoregon.edu}

\author[T.A. Oliynyk]{Todd A. Oliynyk}
\address{School of Mathematical Sciences\\
9 Rainforest Walk\\
Monash University, VIC 3800\\ Australia}
\email{todd.oliynyk@monash.edu}

\begin{abstract}
  We prove the nonlinear stability of the asymptotic behavior of perturbations of subfamilies of Kasner solutions in the contracting time direction within the class of polarised $\Tbb^2$-symmetric solutions of the vacuum Einstein equations with arbitrary cosmological constant $\Lambda$. 
  This stability result generalizes the results proven in \cite{abio2020}, which focus on the $\Lambda=0$ case, and as in that article, the proof relies on an areal time foliation and Fuchsian techniques.   
  Even for $\Lambda=0$, the results established here apply to a wider class of perturbations of Kasner solutions within the family of polarised $\Tbb^2$-symmetric vacuum solutions than those considered in \cite{abio2020} and \cite{fournodavlos2020b}.
  Our results establish that the areal time coordinate takes all values in $(0, T_0]$ for some $T_0 > 0$, for certain families of polarised $\Tbb^2$-symmetric solutions with cosmological constant.  
\end{abstract}

\maketitle

\section{Introduction}
\label{s.introduction}
The \emph{Einstein vacuum equations} with \emph{cosmological constant} $\Lambda$ take the form
\begin{equation}
  \label{eq:EFE}
  G + \Lambda g = 0
\end{equation}
where $G= \Ric + \frac 12 \R g$ is the Einstein tensor of a Lorentzian metric $g$, and $\Ric$ and $\R$ are the Ricci tensor and scalar curvature of $g$, respectively. 
In this article, we focus on \emph{cosmological solutions} $(M,g)$ of \eqref{eq:EFE}, which are characterised by the existence of compact Cauchy hypersurfaces $\Sigma$ in the spacetime manifold $M$. 
Without loss of generality, this allows us to assume that $M$ equals $I\times\Sigma$ where $I$ is some time interval. 
In this article, we restrict our attention to cosmological spacetimes with Cauchy hypersurfaces of the form $\Sigma=\Tbb^3$, and with a $\Tbb^2$ isometry group acting on the spacetime which preserves the Cauchy hypersurfaces.

One of the outstanding problems in mathematical cosmology is to characterise the dynamics of cosmological solutions near \emph{big bang singularities}. 
It is known that the family of spatially homogeneous and isotropic Friedmann-Lema\^itre-Robertson-Walker (FLRW) solutions generically develop curvature singularities after a finite time in the contracting time direction, which are referred to as big bang singularities, both in vacuum and for a wide range of matter models. 
The conjecture that {generic} cosmological solutions develop singularities of big bang type is supported by the \emph{Penrose and Hawking singularity theorems} \cite{hawkingLargeScaleStructure1973} as well as the \emph{BKL conjecture} \cite{belinskii1970,lifshitz1963}. 
In contrast to FLRW dynamics, generic classes of cosmological solutions are expected to exhibit complicated dynamics near big bang singularities where  congruences of timelike observers experience ``chaotic mixmaster'' type oscillatory behaviour \cite{Andersson:2005,Weaver:2001}.

A mathematically rigorous analysis of mixmaster dynamics is an open problem and appears to be currently out of reach.
Fortunately, there is strong evidence that mixmaster behaviour simplifies to less chaotic blow up dynamics, commonly referred to as \emph{asymptotically velocity term dominated behaviour} (AVTD) \cite{Eardley:1972, Isenberg:1990}, in the following three situations: (i) the Einstein's equations are coupled to ``extreme'' matter fields, e.g., minimally coupled scalar fields or stiff fluids \cite{andersson2001,heinzle2012,rodnianski2014, fournodavlos2020b,fournodavlos2021}; (ii) the number of spatial dimensions $D$ is large, i.e., $D\ge 11$ \cite{damour2002}); or (iii) certain symmetries are present \cite{Isenberg:1990,berger1998b,isenberg1999,rendall2000,isenberg2002,stahl2002,choquet-bruhat2004,choquet-bruhat2006,ringstrom2009a,ames2017,beyer2017}. 
In this article, we focus our attention on the third situation and consider \emph{polarised $\Tbb^2$-symmetric classes} of cosmological solutions.

AVTD behaviour is characterised by the existence of a congruence of time-like observers who experience spatially pointwise {Kasner type} behaviour in the contracting time direction. 
While we do not fully establish the existence of solutions with AVTD behaviour in this article%
\footnote{We believe that it should be straightforward to modify the arguments of \cite{abio2020} to establish that AVTD behaviour occurs in the solutions considered in this paper.}, 
we do establish the existence of solutions that exhibit \emph{spatially pointwise Kasner asymptotics} and curvature blow up. 
The main result of this article --- see Theorem~\ref{thm.main_theorem} for a precise statement --- is that we construct open neighborhoods of $\Tbb^2$-symmetric vacuum solutions that contain a range of exact Kasner solutions \cite{kasner1921}. 
Moreover, we establish that these vacuum solutions have the expected dynamics near the singularity, which is quantified via rigorous estimates, and are $C^2$-inextendible due to curvature blow up; i.e., they terminate at a big bang singularity. 
It is interesting to note that, in agreement with heuristic arguments, we find that the cosmological constant does not enter the leading-order estimates near the singularity.
Nevertheless, our results indicate that a non-zero cosmological constant, irrespective of its sign, renders certain otherwise stable Kasner solutions unstable. 
As a consequence, the standard picture of dynamics near the singularity in the presence of a cosmological constant may need to be revised.

We further emphasize that for the $\Tbb^2$-symmetric vacuum solutions – with or without a non-vanishing cosmological constant – the polarization condition that we assume here (see Section~\ref{sec:polarisedT2}) is crucial for AVTD behaviour to occur. 
Numerical studies of $\Tbb^2$-symmetric solutions of the vacuum Einstein equations without the polarization condition and with non-vanishing twist (defined below in Section~\ref{s.T2spacetimes.areal_gauge}) strongly indicate that such solutions generically exhibit mixmaster-like behaviour in a neighborhood of the singularity \cite{Andersson:2005,Weaver:2001} rather than AVTD behaviour.

Two distinct approaches have been employed to investigate the dynamics near the big bang singularity in the mathematical cosmology literature.
The first, an approach based on a \emph{singular initial value problem} for evolution equations in Fuchsian form \cite{ames2013a,beyer2010b,beyer2020d,claudel1998a,kichenassamy1998,kichenassamy2007k,rendall2000}, aims to establish the existence of families of singular solutions of the Einstein-matter or Einstein-vacuum equations  by prescribing asymptotics near the singularity and then solving the equations in the direction away from the singularity. 
While this approach yields infinite-dimensional families of solutions parameterized by free ``asymptotic data'' functions, it does not address whether the family of solutions constructed is open in the set of all solutions under consideration, for example, within a given symmetry class or with certain matter fields. 
To obtain information regarding the behaviour of open sets of solutions, an alternative approach is needed. 
This second approach relies on solving the Cauchy problem for initial data sets specified on a Cauchy hypersurface near the singularity, and studying the behaviour of these solutions as they evolve towards the singularity. 

The second approach, based on evolving initial data from a Cauchy hypersurface towards the singularity, has been used to show that, under the assumption that the cosmological constant vanishes, AVTD behaviour occurs for vacuum solutions of the Einstein equations in the following situations: generically in the class of vacuum Gowdy solutions \cite{CIM1990,Isenberg:1990,ringstrom2009a}, near exact Kasner solutions in the polarised $\Tbb^2$-symmetric setting \cite{abio2020}, and near exact Kasner solutions in the polarised $U(1)$-symmetric setting \cite{fournodavlos2020b}. 
The works \cite{fournodavlos2020b,Rodnianski2018HighD,rodnianski2014} establish  corresponding results for cosmological solutions without symmetries if the Einstein's equations are either coupled to a scalar field or if the number of spatial dimensions is sufficiently large.
On the other hand, using the first approach, infinite dimensional subfamilies of solutions with AVTD behaviour are known in symmetry-defined classes of spacetimes \cite{ames2013a,beyer2017,choquet-bruhat2006,choquet-bruhat2004,Clausen2007,damour2002,isenberg1999,isenberg2002,kichenassamy1998,stahl2002}, in spacetimes without symmetry and coupled to a stiff fluid \cite{andersson2001,heinzle2012}, and in special classes of vacuum spacetimes without symmetry \cite{ChruscielKlinger:2015,Fournodavlos:2020}.
Finally, we mention other recent related work \cite{Lott:2020b,Lott:2020a,ringstrom2017,Ringstrom:2021a,Ringstrom:2021b} that establish geometric conditions under which one obtains detailed information about 
the geometry of certain spacetimes near their singularities.

In this article, we follow the second approach, and in particular, we employ techniques based on the formulation of the Einstein evolution equations in symmetric hyperbolic Fuchsian form to establish the nonlinear stability of Kasner solutions towards the singularity. 
It is important to note that the Fuchsian techniques we use are based on solving a Cauchy problem and not a singular initial value problem, which is where Fuchsian methods have been traditionally employed. 
The use of techniques relying on the Fuchsian formulation of the evolution equations to obtain global existence results for Cauchy problems was initiated in \cite{Oliynyk:CMP_2016}, and since then, this method has been further developed and employed in a variety of settings to establish the global existence of solutions to Cauchy problems for various field theories
\cite{FOW:2021,LeFlochWei:2021,LiuOliynyk:2018b,LiuOliynyk:2018a,LiuWei:2021,Oliynyk:2021,OliynykOlvera:2021}.

The
main result of this article --- Theorem~\ref{thm.main_theorem} below --- is the nonlinear stability of a family of Kasner solutions within the polarised $\Tbb^2$-symmetric spacetimes with cosmological constant. 
Similar results are proved in recent works of \cite{abio2020} and \cite{fournodavlos2020b}.
All of these works, the current one included, establish the dynamical stability of Kasner solutions within classes of cosmological solutions --- polarised $\Tbb^2$-symmetry in this article and \cite{abio2020}, and polarised $U(1)$-symmetry in \cite{fournodavlos2020b}. 
However, the present result generalizes the results of \cite{abio2020,fournodavlos2020b} in the following ways.
First, the present result allows for an arbitrary cosmological constant $\Lambda\in \Rbb$ while \cite{abio2020} and \cite{fournodavlos2020b} are restricted to $\Lambda=0$. 
Second, the present result applies to a wider class of perturbed solutions. 
In the results of both \cite{fournodavlos2020b} and \cite{abio2020}%
\footnote{We note that the perturbed solutions discussed in \cite{abio2020} are a subset of those discussed in \cite{fournodavlos2020b}, while the asymptotic results of \cite{abio2020} are sharper than those of \cite{fournodavlos2020b}.}, 
the perturbed solutions are restricted to those with the hypersurface orthogonal Killing vector field close to parallel to one of the geometrically distinguished eigenvector fields of the given background Kasner solution.
By using a more general family of areal gauges, in the present article we remove this restriction and allow for perturbed solutions which are geometrically distinguished from those in \cite{abio2020,fournodavlos2020b}.
As explained below in Section~\ref{s.T2spacetimes}, the hypersurface orthogonal Killing vector field of the present perturbed solutions may be a linear combination of the eigenvector fields of the background Kasner solution.

\bigskip

\noindent \textit{Overview:} 
In Section~\ref{s.T2spacetimes}, we define the polarised $\Tbb^2$-symmetric cosmological solutions and the areal coordinate gauge. 
A key feature of the areal coordinate gauge is that the time coordinate $t$ synchronizes the singularity at $t=0$, and as $t$ increases, the spacetime expands. 
As a consequence, our analysis of  polarised $\Tbb^2$-symmetric spacetimes near their singularities amounts to analyzing polarised $\Tbb^2$-symmetric vacuum solutions to the Einstein equations in the areal gauge for positive $t$ close to $0$. 
Next, we introduce, in Section~\ref{s.main_result}, the family of Kasner solutions of the vacuum Einstein equations with an arbitrary cosmological constant $\Lambda$ \cite{Garfinkle2021} --- see also Section~\ref{sec:Kasnersol} in the appendix --- and we provide a rigorous statement of our main stability result in Theorem~\ref{thm.main_theorem}. 
 The proof of this theorem is carried out in Section~\ref{s.main_proof}.

\section{$\Tbb^2$-symmetric vacuum spacetimes}
\label{s.T2spacetimes}

\subsection{$\Tbb^2$-symmetry and areal gauge}
\label{s.T2spacetimes.areal_gauge}
A $\Tbb^2$-symmetric cosmological spacetime $(M,g)$ is by definition globally hyperbolic and characterised by the existence of an effective smooth action of 
the isometry group $U(1)\times U(1)\cong\Tbb^2$ \cite{CHRUSCIEL:1990}, which preserves the Cauchy hypersurfaces of the spacetime. 
The corresponding Lie algebra of Killing vector fields is spanned by two commuting spacelike Killing vector fields $X$ and $Y$ with closed orbits. 
In this article, we restrict to spacetime manifolds of the form $M=(0, \infty)\times\Tbb^3$. 

We recall that a $\Tbb^2$-symmetric metric $g$ in \emph{areal coordinates} $(t,\theta,x,y)$ \cite{BERGER1997,CHRUSCIEL:1990} on $(0, \infty)\times\Tbb^3$ takes the
form 
\begin{equation}
\label{T2metric}
    g = e^{2(\nu - u)} \left( -\alpha\rmd t^2 + \rmd \theta^2 \right)
    + e^{2u}\left(\rmd x + Q \rmd y +  (G + QH) \rmd \theta \right)^2
    + e^{-2u} t^2\left(\rmd y + H \rmd \theta\right)^2,
  \end{equation}
where the fields $\nu$, $u$, $Q$, $\alpha$, $G$, $H$ are functions of the coordinates $t$ and $\theta$ only. 
This representation of the metric assumes the following gauge choices: (i) the coordinate vector fields $\partial_x$ and $\partial_y$ are identified with a basis of the Lie algebra of Killing vector fields, that is,
$X=\partial_x$ and $Y=\partial_y$, (ii) the time coordinate $t$ is
\emph{areal}, that is,
\[t^2=\det
  \begin{pmatrix}
    g_{xx} & g_{yx}\\
    g_{xy} & g_{yy}
  \end{pmatrix},\]
 and (iii) the \emph{shift vector field} vanishes identically, which is characterised by
 \[g_{t\theta}=g_{tx}=g_{ty}=0.\]

In general, the areal coordinates are only defined locally, and in particular, for $t$ on finite intervals of the form $(t_0,t_1)$ with $0< t_0 < t_1$. However, if $\Lambda=0$, it has been shown that solutions of \eqref{eq:EFE} in areal coordinates can always be extended to solutions on the whole interval $t\in (0,\infty)$ \cite{BERGER1997,IsenbergWeaver:2003}. 
In the nonpolarised case (see definition of polarised below), Smulevici has shown \cite[Thm.~2]{Smulevici:2011} that $\Tbb^2$-symmetric spacetimes with positive cosmological constant  are covered by an areal foliation with $t\in (0,\infty)$. 
If $\Lambda \neq 0$ and the spacetime is polarised, there are no general results in the literature for vacuum spacetimes regarding the lower bound $t_0$ of the areal time coordinate. 
It is interesting then to note that we identify, in this article, a class of polarised $\Tbb^2$-symmetric solutions of \eqref{eq:EFE} with arbitrary $\Lambda\in\Rbb$ on intervals of the form $t\in (0,T_0]$ for some $T_0>0$. 

To each of the two Killing vector fields $X$ and $Y$ of $g$, there corresponds a twist one-form defined by
\begin{equation}
  \label{eq:twistoneforms}
  \omega^{(X)}_d=\epsilon_{abcd}(\nabla^a X^b)X^c \AND
  \omega^{(Y)}_d=\epsilon_{abcd}(\nabla^a Y^b)Y^c,
\end{equation}
respectively, where $\epsilon_{abcd}$ is the fully skew-symmetric Levi-Civita symbol.
Given these, the scalar fields
\begin{align}
  \label{eq:twistconstantX}
  J_{X} =& \omega^{(X)}_dY^d
           = -t e^{4u - 2\nu} \alpha^{-1/2} (\del{t} G + Q \del{t}H), \\
  \label{eq:twistconstantY}
  J_{Y} =& -\omega^{(Y)}_dX^d
        = -t e^{-2\nu} \alpha^{-1/2} (t^2 \del{t} H + Q e^{4u}(\del{t} G + Q \del{t}H))
        = QJ_{X} - e^{-2\nu} \alpha^{-1/2} t^3 \del{t}H,
  \end{align} 
  are defined as the \emph{twist constants} of $g$. 
  These names are motivated by the fact \cite{CHRUSCIEL:1990} that $J_X$ and $J_Y$ are constant for $\Tbb^2$-symmetric solutions of the vacuum Einstein equations \eqref{eq:EFE} with arbitrary cosmological constant.
It can also be shown that the twist one-forms are given by
  \begin{align}
    \label{eq:twistoneforms1}
    \omega^{(X)}=&
                  \frac{\del{\theta}Q \sqrt{\alpha } e^{4 u}}{t} \rmd t
                  +\left({J_X} H+\frac{\del{t}Q e^{4
                  u}}{t \sqrt{\alpha }}\right)\rmd\theta
                   + {J_X} \rmd y
                = \frac{\sqrt{\alpha}}{t}e^{2u+2\nu}i_{\grad Q}\rmd t\wedge \rmd\theta
                  + J_{X}(\rmd y + H \rmd \theta)
  \end{align}
  \begin{align}                  
    \begin{split}                  
    \label{eq:twistoneforms2}
    \omega^{(Y)}=&
                  -\frac{\sqrt{\alpha } \left(\del{\theta}Q \left(t^2-Q^2 e^{4 u}\right)+4 t^2 Q
                  \del{\theta}u\right)}{t} \rmd t \\
                 &+\frac{\del{t}Q Q^2 e^{4 u}+2 t Q \left(1-2 t \del{t}u\right) -t \left({J_Y} G \sqrt{\alpha }+t \del{t}Q(t,\theta
                  )\right)}{t \sqrt{\alpha
                  }} \rmd\theta
                  -{J_Y} \rmd x \\
                =& \frac{-\sqrt{\alpha}}{t}e^{2\nu - 2u}
                  \biggl((e^{4u}Q^2 -t^2) i_{\grad Q}\rmd t\wedge\rmd \theta
                      - 4t^2Qi_{\grad u}\rmd t\wedge\rmd \theta
                  \biggr) 
                  - J_{Y}(\rmd x + G\rmd \theta)
                  + \frac{2Q}{\sqrt{\alpha}}\rmd \theta,
    \end{split}                  
  \end{align}
 where $i_{\grad F}\rmd t\wedge\rmd \theta = -\alpha^{-1}e^{2u-2\nu} (\alpha\del{\theta}F \rmd t + \del{t}F \rmd \theta)$ for any chosen function $F(t,\theta)$. 

 \newcommand{\sout}{{\tilde t}}
 \newcommand{\southeta}{{\thetat}}
 \newcommand{\soux}{\xt}
 \newcommand{\souy}{\yt}
 \newcommand{\sounu}{\nut}
 \newcommand{\souu}{\ut}
 \newcommand{\soualpha}{\alphat}
 \newcommand{\souQ}{\Qt}
 \newcommand{\souG}{\Gt}
 \newcommand{\souH}{\Ht}
 \newcommand{\tart}{{\tilde t}}
 \newcommand{\tartheta}{{\thetat}}
 \newcommand{\tarx}{\xt}
 \newcommand{\tary}{\yt}
 \newcommand{\tarnu}{\nut}
 \newcommand{\taru}{\ut}
 \newcommand{\tarQ}{\Qt}
 \newcommand{\taralpha}{\alphat}
 \newcommand{\tarG}{\Gt}
 \newcommand{\tarH}{\Ht}

It is important to note that the form \eqref{T2metric} of a $\Tbb^2$-symmetric metric in areal coordinates is \emph{invariant} under the following family of coordinate transformations
\begin{equation}
  \label{eq:coordtrafo}
  (\tilde t,\thetat,\xt,\yt)=(a \,t,\theta, a_{11}x+a_{12}y,
  a_{21}x+a_{22}y), \quad a=|\det A|,
\end{equation}
for any constant matrix
\begin{equation}
  \label{eq:coordtrafo2}
  A=
  \begin{pmatrix}
    a_{11} & a_{12}\\
    a_{21} & a_{22}
  \end{pmatrix}\in GL(2,\Rbb).
\end{equation}

By \emph{invariant}, we mean that the metric \eqref{T2metric}, if written in terms of the coordinates \eqref{eq:coordtrafo}, takes on the same form; that is, 
\begin{equation}
\label{T2metrictafo}
  {g} = e^{2(\nut - \ut)} \left( -\alpha\rmd {\tilde t}^2 + \rmd      \thetat^2 \right)
    + e^{2\ut}\left(\rmd \xt + \Qt \rmd \yt +  (\Gt + \Qt \Ht) \rmd \thetat \right)^2
    + e^{-2\ut} {\tilde t}^2\left(\rmd \yt + \Ht \rmd \thetat\right)^2.
  \end{equation}
For use below, we note, as can be verified by a straightforward calculation, that the metric function $Q$
transforms as
\begin{equation}
  \label{eq:Qtransform}
Q(t,\theta)=\frac{(a_{11}+a_{12}\Qt(\tilde t,\thetat)) (a_{21}+a_{22}\Qt(\tilde t,\thetat)) +a_{12}a_{22} e^{-4\ut(\tilde t,\thetat)}{\tilde t}^2}
  {(a_{11}+a_{12}\Qt(\tilde t,\thetat))^2 +a_{12}^2 e^{-4\ut(\tilde t,\thetat)}{\tilde t}^2},
\end{equation}
while the Killing bases associated to the two coordinate systems, i.e., $\{\Xt=\del{\xt},\Yt=\del{\yt}\}$ and $\{X=\del{x}, Y=\del{y}\}$, and their corresponding twist constants
are related via
\begin{equation}
  \label{eq:KVFTrafo}
  \begin{pmatrix}
    X\\ Y
  \end{pmatrix}=A \begin{pmatrix}
    \Xt\\ \Yt
  \end{pmatrix}  
\end{equation}
and
  \begin{equation}
    \label{eq:Jtransform}    
  \begin{pmatrix}
    J_{\Xt}\\ J_{\Yt}
  \end{pmatrix}=(\det A) A \begin{pmatrix}
    J_X\\ J_Y
  \end{pmatrix}.
\end{equation}

\subsection{Polarised $\Tbb^2$-symmetric vacuum spacetimes\label{sec:polarisedT2}}
A $\Tbb^2$-symmetric metric \eqref{T2metrictafo} is called \emph{polarised} if there exists a coordinate transformation \eqref{eq:coordtrafo}-\eqref{eq:coordtrafo2} that results in $Q$, given by \eqref{eq:Qtransform}, vanishing.
For such metrics, the Killing vector field basis $\{X,Y\}$ defined by \eqref{eq:KVFTrafo} is orthogonal everywhere on $M$. 
Now, if a polarised $\Tbb^2$-symmetric metric $g$ defines a solution of the vacuum Einstein equations \eqref{eq:EFE} and is in the form \eqref{T2metric} with $Q=0$, which is always possible by definition, then it can be shown that the twist constants \eqref{eq:twistconstantX}-\eqref{eq:twistconstantY} associated to the Killing basis $\{X,Y\}$ have the property that at least one of them must be zero.
By applying another transformation \eqref{eq:coordtrafo}-\eqref{eq:coordtrafo2} with $a_{11}=0$, $a_{21}=a_{12}=1$ and $a_{22}=0$ if necessary, it is always possible to transform the metric into the form \eqref{T2metric} with $Q=0$, $J_X=0$ and $J_Y\in\Rbb$. We call this  the \emph{$Q_0$-areal gauge} for polarised $\Tbb^2$-symmetric solutions of vacuum Einstein equations \eqref{eq:EFE}. We also note that if the twist constants happen to satisfy  $J_X=J_Y=0$, then the metric $g$ is referred to as a \emph{polarised Gowdy metric} \cite{CIM1990,CHRUSCIEL:1990,Gowdy1974,Isenberg:1990}.

While most investigations of polarised $\Tbb^2$-symmetric solutions of the Einstein equations in the literature \cite{isenberg1999,Clausen2007,ames2013a,abio2020} employ the $Q_0$-areal gauge, we find that this gauge is too restrictive for our purposes, and so we consider here more general gauges.
Given a polarised $\Tbb^2$-symmetric metric $g$ in $Q_0$-areal gauge, i.e., \eqref{T2metrictafo} with $\Qt=0$, the \emph{class of all areal coordinate representations} of $g$ is constructed from the class of all coordinate transformations \eqref{eq:coordtrafo}- \eqref{eq:coordtrafo2}.
The corresponding metric functions $Q$ are then of the form
\begin{equation}
  \label{eq:QtransformPol}
  Q(t,\theta)=\frac{a_{11}a_{21}+a_{12}a_{22} e^{-4\ut(\tilde t,\thetat)}{\tilde t}^2}
  {a_{11}^2 +a_{12}^2 e^{-4\ut(\tilde t,\thetat)}{\tilde t}^2},
\end{equation}
which follows from setting $\Qt=0$ in \eqref{eq:Qtransform}.

In the following, we need not consider the whole class of areal coordinate representations, and instead, we find it sufficient to restrict to the special class for which $Q$ is \emph{constant} in both space and time but \emph{not necessarily zero}; we refer to these coordinate systems as \emph{$Q_{\text{const}}$-areal gauges}. 
In fact, we further specialise to $Q_{\text{const}}$-areal gauges that are obtained from
a $Q_0$-areal gauge in one of the following two ways:
\begin{enumerate}
\item [(1):] Choose $Q\in\Rbb$, fix arbitrary non-vanishing $a_{12}$, $a_{21}$ in $\Rbb$, and set
\[
  A=
  \begin{pmatrix}
      0 & a_{12}\\
      a_{21} & Q a_{12}\\
  \end{pmatrix}.
\]
Then it follows from \eqref{eq:Jtransform} that
\begin{equation}
  \label{eq:TwistConstants.1}
    J_{X}\in\Rbb \AND J_{Y}=Q  J_{X},
  \end{equation}
and moreover, it can be shown, for $J_X\neq 0$, that
\begin{equation}  
  \label{eq:hypersurfaceorthogonalKVF.1}
  Y-Q X
\end{equation}
is the unique,  up to rescaling, hypersurface orthogonal Killing vector field of $g$. 
\item [(2):] Choose $Q\in\Rbb$, fix arbitrary non-vanishing $a_{11}$, $a_{22}$ in $\Rbb$, and set
\[A=
    \begin{pmatrix}
      a_{11} & 0\\
      Q a_{11}& a_{22}\\
    \end{pmatrix}.
\]
In this case, it follows from  \eqref{eq:Jtransform}  that
\begin{equation}
  \label{eq:TwistConstants.2}
    J_{X}=0 \AND J_{Y}\in\Rbb,
  \end{equation}
and it can be shown, for $J_Y\neq 0$, that
\begin{equation}
  \label{eq:hypersurfaceorthogonalKVF.2}
X
 \end{equation}
is the unique, up to rescaling, hypersurface orthogonal Killing vector field of $g$.
\end{enumerate}
These two $Q_{\text{const}}$-areal gauges yield polarised $\Tbb^2$-symmetric solutions of the vacuum Einstein equation \eqref{eq:EFE} of the form \eqref{T2metric} where $Q$ is constant and the twist constants \eqref{eq:twistconstantX} and \eqref{eq:twistconstantY} have the property
\begin{equation}
  \label{eq:twistconstgaugeres}
  \text{$J_X=0$ and  $P\in\Rbb$,} \quad \text{or}\quad \text{$P=0$ and  $J_X\in\Rbb$},\quad \text{or} \quad J_X=P=0
\end{equation}
where $P$ is defined by
\begin{equation}
  P:=J_Y-QJ_X.
\end{equation}
In the $Q_0$-areal gauge used in \cite{abio2020}, in which $Q= J_X = 0$ and $X$ is the hypersurface orthogonal Killing vector field, the transformation (2) above is a simply a rescaling, while transformation (1) flips the coordinate basis Killing vector fields. 
In this gauge, both transformations result in exactly one of the coordinate basis Killing vector fields being hypersurface orthogonal (provided there is one non-vanishing twist). 
By generalizing to the $Q_{\text{const}}$-areal gauges we allow spacetimes in which the unique hypersurface orthogonal Killing vector field is not aligned with one of the coordinate basis Killing vector fields, and hence with an eigenvector field of the background geometry.
The case $J_X=P=0$ corresponds to the well-studied polarised Gowdy metrics.

\subsection{The Einstein vacuum equations for polarised $\Tbb^2$-symmetric spacetimes}
For arbitrary $\Lambda\in\Rbb$, $Q\in\Rbb$, $J_X\in\Rbb$ and $P\in\Rbb$, it is straightforward to show that the Einstein vacuum equations \eqref{eq:EFE} are equivalent to the following system of equations for a polarised $\Tbb^2$-symmetric metrics in a $Q_{\text{const}}$-areal gauge:
  \begin{align}
    \label{eq:vacevolFirst}
    \del{tt} u +\frac{\del{t} u}{t} -\alpha \del{\theta\theta}u
    &=\frac{1}{2} \del{\theta}\alpha \del{\theta} u
    -\frac{{J_X^2} \alpha  \left(t \del{t} u-1\right) e^{2 \nu -4 u}}{2 t^2}
    -\frac{P^2 \alpha  e^{2 \nu } \del{t} u}{2 t^3}
    -\Lambda  \alpha \left(2 t \del{t} u-1\right) e^{2 \nu -2 u},\\
    \del{t}\nu
    &=t\alpha  (\del{\theta} u)^2+t(\del{t} u)^2
    +\frac{{J_X^2} \alpha  e^{2 \nu -4 u}}{4 t}
    +\frac{P^2 \alpha e^{2 \nu }}{4 t^3}    
      +\Lambda  t \alpha  e^{2 \nu -2 u},\\
    \label{eq:vacevolalpha}
    \del{t}\alpha
    &=
    -\frac{{J_X^2} \alpha ^2 e^{2 \nu -4 u}}{t}
    -\frac{P^2 \alpha ^2 e^{2 \nu }}{t^3}
    -4 \Lambda  t \alpha ^2 e^{2 \nu -2 u},\\
    \partial_t G
    &=-\frac{{J_X} \sqrt{\alpha } e^{2 \nu -4 u}}{t}
    +\frac{P Q \sqrt{\alpha }
    e^{2 \nu }}{t^3},\\
    \partial_t H&=
    -\frac{P \sqrt{\alpha } e^{2 \nu }}{t^3}, \label{eq:vacevolH}
  \end{align}
  and
  \begin{align}
    \label{eq:vacconstraint}
    2 t \del{\theta} u \del{t} u-\frac{\del{\theta}\alpha}{2\alpha }-\del{\theta}\nu&=0,\\
    \label{eq:vacevolLast}
    P {J_X} &=0.
  \end{align}
These equations are naturally viewed as evolution and constraint equations determined by \eqref{eq:vacevolFirst}-\eqref{eq:vacevolH} and \eqref{eq:vacconstraint}-\eqref{eq:vacevolLast}, respectively. 
We note that the constraint equation \eqref{eq:vacevolLast} is equivalent to \eqref{eq:twistconstgaugeres}.

\section{Nonlinear stability of $\Lambda$-Kasner solutions}
\label{s.main_result}

The main result of our article --- see Theorem~\ref{thm.main_theorem} for a precise statement --- is the nonlinear stability in the contracting time direction of polarised $\Tbb^2$-symmetric perturbations of the family of \emph{$\Lambda$-Kasner solutions} to the vacuum Einstein equations \eqref{eq:EFE} with non-vanishing cosmological constant $\Lambda$.
The family of $\Lambda$-Kasner solutions \cite{Garfinkle2021} is discussed in detail in Appendix~\ref{sec:Kasnersol}. 
Briefly, it is comprised of spatially homogeneous solutions of \eqref{eq:EFE}, which in Gaussian coordinates $(\tb,\xb_1,\xb_2,\xb_3)$ take the form
\begin{equation}
  \label{eq:KasnermetricGauss}
  g^{(K)}=-\rmd \tb^2+\LKone^{2/3}\bigl(z(\tb)\bigr)\sum_{i=1}^3 \tb^{2p_i}\LKtwo^{2p_i}\bigl(z(\tb)\bigr)d\xb_i^2,
\end{equation}
where the \emph{Kasner exponents} $p_1,p_2,p_3\in\Rbb$ satisfy the \emph{Kasner relations}
\begin{equation}
  \label{eq:Kasner}
  p_1 + p_2 + p_3 = p_1^2 + p_2^2 + p_3^2 = 1.
\end{equation}
The complex analytic functions $\LKone(z)$ and $\LKtwo(z)$ are specified in equation \eqref{eq:K1K2}, and it follows from their definitions that $\LKone(0)=\LKtwo(0)=1$, and that $\LKone(z)$ and $\LKtwo(z)$ are real-valued for all $z\in\Rbb$ and all $z\in i\Rbb$, where, for arbitrary $\Lambda\in\Rbb$, $z(\tb)$ is defined by
\begin{equation}
  \label{eq:defz}
  z(\tb)=\frac{\sqrt{3}}2 \sqrt{\Lambda}\, \tb
\end{equation}
with $\sqrt{\cdot}$ an arbitrary choice of branch of the complex square root function.

As we show in Appendix~\ref{sec:Kasnersol}, we can introduce areal coordinates in which the $\Lambda$-Kasner metrics \eqref{eq:KasnermetricGauss} take the form \eqref{T2metric} with
\begin{align}
  \label{eq:KasnersolFirst}
  u^{(K)} &= \frac{1-K}2 \ln(t) +\Ord\left(\Lambda\, t^{(K^2+3)/2}\right),\\
  {\nu^{(K)}}&=\frac{(1-K)^2}4 \ln(t) +\Ord\left(\Lambda\, t^{(K^2+3)/2}\right),\\
  \label{eq:Kasnersolalpha}
  \alpha^{(K)}&=1+\Ord\left(\Lambda\, t^{(K^2+3)/2}\right),\\
  \label{eq:KasnersolLast}
  Q^{(K)}&=G^{(K)}=H^{(K)}=0,
\end{align}
see especially \eqref{eq:Kasnerexp1} -- \eqref{eq:Kasnerexp3},
where the areal time $t$ is related to the Kasner time $\tb$ by \eqref{eq:KasnerCoordTrafo} and
$K\in\Rbb$ is an arbitrary parameter that determines the Kasner exponents according to
\begin{equation}
\label{eq:Kasnerexpo}
 p_1 =({K}^2-1)/({K}^2+3), \quad
 p_2 =2(1+{K})/({K}^2+3),\quad
 p_3 =2(1-{K})/({K}^2+3).
\end{equation}
Observe that \eqref{eq:Kasnerexpo} implies \eqref{eq:Kasner} for all choices of $K\in\Rbb$. 
It is also straightforward to see from \eqref{eq:KasnersolFirst} -- \eqref{eq:KasnersolLast} that the asymptotics of the $\Lambda$-Kasner solutions as  $t\searrow 0$ are identical to those of the standard vacuum \emph{Kasner solutions} \cite{kasner1921,abio2020}  given by setting $\Lambda=0$.
The parameters $\Lambda\in \Rbb$ and $K\in \Rbb$, through the relations \eqref{eq:defz} and \eqref{eq:Kasnerexpo}, define the 2-parameter family of spatially homogeneous $\Lambda$-Kasner solutions.

For a given $\Lambda$-Kasner solution $g^{(K)}$ with $K,\Lambda \in\Rbb$ expressed in areal coordinates as in \eqref{eq:KasnersolFirst}-\eqref{eq:KasnersolLast} above, the coordinate vector fields $\{\partial_\theta,\partial_x,\partial_y\}$ carry geometric information because, other than in exceptional cases, they each span one of the three $1$-dimensional eigenspaces of the  Weingarten map  $\Kbb^{(K)}$ induced by $g^{(K)}$ on the $t=const$-surfaces%
\footnote{
  Recall from the discussion in Appendix~\ref{sec:Kasnersol} that the two foliations of a $\Lambda$-Kasner spacetime induced by the two functions $t$ and $\tb$ are identical and so are the respective Weingarten maps. 
  Moreover, the vector fields $\partial_\theta$, $\partial_x$, $\partial_y$ are parallel to $\partial_{\xb_1}$, $\partial_{\xb_3}$, $\partial_{\xb_2}$, respectively.
} 
The exceptional cases are if $K=0$, $K=\pm 1$, or $K=\pm 3$ in which case some of these eigenspaces have multiplicity two.  
In particular, the vector fields $X=\partial_x$ and $Y=\partial_y$  carry distinct geometric information regarding eigenspaces of the Weingarten map.

The  class of polarised  $\Tbb^2$-symmetric perturbations of the  $\Lambda$-Kasner solutions $g^{(K)}$ for which we establish stability in the contracting time direction can be informally defined as the set of solutions $g$ of the vacuum Einstein equations \eqref{eq:EFE} that are sufficiently close to $g^{(K)}$ in a given $Q_\text{const}$-areal gauge. Moreover, we allow for arbitrary $Q, J_X,P\in\Rbb$ with $J_X P=0$, noting that the \emph{Killing vector fields $\partial_x$ and $\partial_y$ of $g$ agree with the geometrically distinguished eigenvector fields $X$ and $Y$ of the Weingarten map of the $\Lambda$-Kasner solution $g^{(K)}$}.
As a consequence of the fact that we work here with a number of spacetime geometries on the manifold M – those corresponding to the background solutions $g^{(K)}$ as well as those corresponding to the perturbed solutions $g$ – for this perturbation analysis, the quantity $Q$, which in \cite{abio2020} has been described as ``nothing but gauge", now becomes the carrier of geometric information.
As we explain below, of particular importance is that the unique hypersurface orthogonal Killing vector field of $g$ is parametrised by $Q$ according to \eqref{eq:hypersurfaceorthogonalKVF.1} and \eqref{eq:hypersurfaceorthogonalKVF.2}.
Choosing the   $Q_\text{const}$-areal gauge as opposed to the $Q_\text{0}$-areal gauge for the study of perturbations of $\Lambda$-Kasner solutions is therefore not just a matter of convenience; rather it leads to a new class of perturbations, the consequences of which are discussed below. 

\begin{thm}
  \label{thm.main_theorem}  
  Suppose $k\in\Zbb_{\ge 3}$, and the constants $Q,\Lambda, J_X, P, K \in \Rbb$
  satisfy $P J_X=0$ and one of the following: 
  \begin{align}
    \label{eq.JXnonzero.1}
    &\text{Case $1$ $\quad J_X \ne 0$, $P=0$, and $3 < K$},\\
    \label{eq.JXzero_Lambdanonzero.1}
    &\text{Case $2$ $\quad J_X=0$, $\Lambda \ne 0$,  
    and $K\in (-3,-1)\cup (3,\infty)$},\\
    \label{eq.JXzero_Lambdazero.1}
    &\text{Case $3$ $\quad J_X=\Lambda = 0$, and $K\in (-\infty,-1)\cup (3,\infty)$.}
  \end{align}  
Additionally, corresponding to the above cases, fix constants
$\ell$ and $\kappa_0$ according to 
\begin{align}
  \label{eq:kappa0final}
  \begin{split}
    &\text{Case $1$} \quad  \ell =\frac 12, \quad \kappa_0 = \min\{1,(K-3)(K+1)/4\}, \\
    &\text{Case $2$} \quad  \ell = 2,       \quad \kappa_0 = \min\{1,(K-3)(K+1)/4,(3+K)/2\}, \\
    &\text{Case $3$} \quad  \ell =\infty,   \quad \kappa_0 = \min\{1,(K-3)(K+1)/4\}.
  \end{split}
\end{align}
Then for every sufficiently small $T_0>0$,  there exist constants $R_0>0$ and $\sigma\in (0,2\kappa_0/3)$, such that, for each choice of 
Cauchy data $(\mathring u, \mathring v, \mathring\nu, \mathring\alpha,\mathring G,\mathring H)\in H^k$
satisfying the constraints \eqref{eq:vacconstraint} and \eqref{eq:vacevolLast} at $t=T_0$ and
\begin{equation}
    \label{eq:PerturbedKasnerCauchyData}
    \Bnorm{\Bigl(T_0^\ell e^{-\mathring{u}}, T_0\mathring v-(1-K)/2, T_0\del{\theta} \mathring{u},\mathring{\alpha}-1,T_0\del{\theta}{\mathring\alpha},T_0^{-1}e^{\mathring{\nu}}\Bigr)}_{H^k}<R_0,
  \end{equation}
  there exist functions
  \begin{equation}
    \label{eq:thmreg1}
  u,\nu,\alpha,G,H \in C^0((0,T_0],H^k(\Tbb))\cap C^1((0,T_0],H^{k-1}(\Tbb))
  \end{equation}
  that define a unique classical solution on $(0,T_0]\times \Tbb$ of the polarised $\Tbb^2$-symmetric vacuum Einstein equations \eqref{eq:vacevolFirst}-\eqref{eq:vacevolLast} and the initial conditions
  \begin{equation}
    (u,\partial_t u,\nu,\alpha,G,H)|_{t=T_0}=(\mathring u,\mathring v,\mathring \nu,\mathring \alpha,\mathring G,\mathring H).
  \end{equation}
  Moreover, 
  \begin{enumerate}[(a)]
  \item $u$ satisfies the additional regularity condition
    \begin{equation}
      \label{eq:thmreg2}
      \del{t}u, \del{\theta}u \in C^0((0,T_0],H^k(\Tbb))\cap C^1((0,T_0],H^{k-1}(\Tbb)),
  \end{equation}
  \item there exist functions $\Kt,\ut,\nut,\alphat \in H^{k-1}(\Tbb)$ satisfying
  \begin{equation}
    \label{eq:smallnessKalpha}
    \norm{K-\Kt}_{H^{k-1}}+\norm{1-\alphat}_{H^{k-1}}\lesssim R_0
  \end{equation}
  and
  \begin{equation}
    \label{eq.AsymptoticConstraintFinal}
  \partial_\theta \nut
  + (1-\Kt) \partial_\theta \ut
  + \frac 12 \alphat^{-1} \partial_\theta\alphat=0        
\end{equation}
  such that
  \begin{gather}
    \label{eq:ualphaest}
    \norm{\alpha-\alphat}_{H^{k-1}}
      + \norm{t\partial_t u-\frac{(1-\Kt)}{2}}_{H^{k-1}}
      \lesssim t + t^{2\kappa_0-2\sigma},\\
    \label{eq:unuest}
    \norm{u-\frac{(1-\Kt)}{2}\ln(t)-\ut}_{H^{k-1}}
      +\norm{\nu-\frac{(1-\Kt)^2}{4}\ln(t)-\nut}_{H^{k-1}}
      \lesssim  t+t^{2\kappa_0-2\sigma},\\
    \label{eq:GLimit.Final}
    \Bnorm{G(t)-\Gt}_{H^{k-1}}\lesssim |QP| t^{\min_{\theta\in(0,2\pi]}\{(1-\Kt)^2/2\}-2-\sigma} +| J_X| t^{\min_{\theta\in(0,2\pi]}\{(1+\Kt)^2\}/2-2-\sigma},\\
    \label{eq:HLimit.Final}
    \Bnorm{H(t)-\Ht}_{H^{k-1}}\lesssim |P| t^{\min_{\theta\in(0,2\pi]}\{(1-\Kt)^2\}/2-2-\sigma},
  \end{gather}
  for all $t\in (0,T_0]$, and 
  \item the Kretschmann scalar $\Ic$ diverges according to
\begin{equation}
  \label{eq.KretschmannLeadingOrder}
  \lim_{t\searrow 0}\Bnorm{\Ic\, t^{\Kt^2+3} - \frac{(\Kt^2 + 3)(\Kt^2-1)^2}{4\alphat^2} e^{-4(\ut + \nut)}}_{H^{k-3}}=0.
\end{equation}
\end{enumerate}
\end{thm}

\noindent 
The proof of this theorem is given in Section~\ref{s.main_proof}. Before proceeding with the proof, we first make a number of observations and remarks.

By assumption, the metric coefficients of the $\Lambda$-Kasner solution $g^{(K)}$  determined by  $\Lambda$ and $K$  correspond to the areal gauge \eqref{T2metric}, with the asymptotic behaviour of the metric coefficients given by \eqref{eq:KasnersolFirst}-\eqref{eq:KasnersolLast} (see Section~\ref{sec:Kasnersol} for details) as $t\searrow 0$. 
The solution $g$ to the vacuum Einstein equations \eqref{eq:vacevolFirst}-\eqref{eq:vacevolLast} corresponding to the functions $(u,\nu,\alpha,G,H)$ discussed in this theorem can then be interpreted as a polarised $\Tbb^2$-symmetric perturbation of $g^{(K)}$ written in terms of a fixed $Q_{\text{const}}$-areal gauge. 
The precise meaning of ``perturbation'' here is determined by the initial data smallness condition  \eqref{eq:PerturbedKasnerCauchyData}.
Consistent with \eqref{eq:KasnersolFirst}-\eqref{eq:Kasnersolalpha} as well as the restrictions for $K$ and $\ell$ given by \eqref{eq.JXnonzero.1}-\eqref{eq:kappa0final}, this initial data condition implies that $t^\ell e^{-u}$, $t\partial_tu-(1-K)/2$, $t\partial_\theta u$, $\alpha-1$, $t\partial_\theta\alpha$ and $t^{-1} e^\nu$
are required to be small at the initial time $t=T_0$, where we notice especially that $t\partial_tu-(1-K)/2$ is approximately $t\partial_tu-t\partial_tu^{(K)}$ and $\alpha-1$ is approximately $\alpha-\alpha^{(K)}$ at $t=T_0$ when $T_0$ is small. The reason for the particular form \eqref{eq:PerturbedKasnerCauchyData} of this smallness condition becomes evident in the proof of Theorem~\ref{thm.main_theorem} discussed in Section~\ref{s.main_proof}.
It is important to note, in any case, that the initial values of $G$ and $H$ as well as the values of the parameters $Q$, $J_X$ and $J_Y$ are \emph{not} required to be small.

Theorem~\ref{thm.main_theorem} states that any such perturbation $g$ extends with the regularity given by \eqref{eq:thmreg1} and \eqref{eq:thmreg2} down to $t=0$ in areal coordinates (i.e., it does not become singular before the time $t=0$ is reached). 
Due to  \eqref{eq.KretschmannLeadingOrder}, we know that these perturbation solutions terminate at a big bang curvature singularity at $t=0$, and as a consequence, they are $C^2$-inextendible past the singularity. 
Precise dynamical information about the limit at $t=0$ is provided by the  estimates \eqref{eq:ualphaest}-\eqref{eq:HLimit.Final} for the metric component functions of the perturbation. 
Moreover, if $Q$, $J_X$, $J_Y$, $\mathring G$ and $\mathring H$ are small, which would imply that the perturbation being consider is genuinely a small perturbation of $g^{(K)}$, we can interpret these estimates as the statement that perturbations converge to a new ``$\Lambda$-Kasner metric'' $g^{(\Kt)}$ at $t=0$ with a new ``Kasner parameter'' $\Kt$ at $t=0$. 
However, since $\Kt$ can be spatially dependent, the limit metric $g^{(\Kt)}$ is, in general, \emph{not} a solution of the vacuum Einstein equations.

It is of particular interest to note that Theorem~\ref{thm.main_theorem} implies the nonlinear stability in the contracting direction of the subfamily of $\Lambda$-Kasner solutions determined by the choice of parameters according to one of the three cases \eqref{eq.JXnonzero.1}-\eqref{eq.JXzero_Lambdazero.1}.

\subsection*{Case 3} 
In this case, $J_X=\Lambda=0$ and $J_Y\in\Rbb$, and  according to \eqref{eq:TwistConstants.2}-\eqref{eq:hypersurfaceorthogonalKVF.2}, the unique hypersurface orthogonal Killing vector field of $g$ agrees with the eigenvector field $X$ of the Weingarten map associated with the foliation of $t=const$-surfaces of $g^{(K)}$. 
For $Q=0$, Theorem~\ref{thm.main_theorem} then reduces to Theorem~4.1 in \cite{abio2020} with $m=J_Y^2$. In \cite{abio2020}, we provide examples of perturbations of Kasner solutions whose Kasner parameters $K$ violate \eqref{eq.JXzero_Lambdazero.1} and exhibit unstable dynamics near $t=0$. 
Based on this, we conjecture that the restriction \eqref{eq.JXzero_Lambdazero.1} for $K$ is sharp. 
The Case~3 of Theorem~\ref{thm.main_theorem} with $Q\not=0$ yields a small extension of the results in \cite{abio2020} since for such solutions  the two Killing vector fields $X$ and $Y$ of $g$ are not orthogonal. 
Similar to \cite{abio2020}, however, we expect the allowed range of $K$ \eqref{eq.JXzero_Lambdazero.1}, which is notably unaffected by the presence of $Q$, to be sharp since it should be possible to generalize the examples from \cite{abio2020} that exhibit unstable dynamics to this new setting.

\subsection*{Case~2} This case is similar to Case~3, but now allows $\Lambda$ to be an arbitrary non-zero constant. This yields a new stability result that is not covered by the results of \cite{abio2020}, irrespective of the value of $Q$. It is worth noting that the presence of the cosmological constant reduces the range \eqref{eq.JXzero_Lambdanonzero.1} of  $K$-values for which stability is guaranteed. In the same way as for Case~3, we expect this restriction for $K$ to be sharp and that Kasner solutions outside this range are unstable near $t=0$.

\subsection*{Case~1} 
This case, where $J_X,Q,\Lambda\in\Rbb$ are unrestricted and $J_Y=Q J_X$ must hold in general as a consequence of \eqref{eq:vacevolLast}, has led to the most significant new stability results that we can establish in this article. 
The significance of this case is due to the fact that the unique hypersurface orthogonal Killing vector field of $g$ is $Y-Q X$ --- see \eqref{eq:TwistConstants.1}-\eqref{eq:hypersurfaceorthogonalKVF.1} --- can be an arbitrary linear combination of the two eigenvector fields $X$ and $Y$. 
As we discuss above this means that the class of perturbations in Case~1 is not included in the otherwise more general classes of perturbations discussed in \cite{fournodavlos2020b, fournodavlos2021}. 
It is worth noting that due to \eqref{eq.JXnonzero.1} the range of $K$-values for which our stability result holds in this case is reduced even further in comparison to Cases~2~and~3.  
As for Cases~2~and~3, we expect this restriction for $K$ to be sharp and that Kasner solutions outside this range are unstable near $t=0$.

\bigskip

As mentioned above we consider Case~1 \eqref{eq.JXnonzero.1} as the main new contribution of this paper. 
This is because if $Q$ is sufficiently large this class is not contained in the family of ($U(1)$-symmetric) perturbations of Kasner solutions studied in \cite{fournodavlos2020b}, even if $\Lambda=0$; this is in contrast with our stability results in \cite{abio2020}, which are covered by the work in \cite{fournodavlos2020b}.
The reason that the stability results established in this paper are not covered by those from \cite{fournodavlos2020b} is because in that article, the perturbations are required to have a hypersurface orthogonal Killing vector field that is close to one of the distinct eigenvector fields, say $X$, of $g^{(K)}$. 
As discussed above, the hypersurface orthogonal vector field of our perturbations $g$ in Case~1 is in general an arbitrary linear combination of the distinct eigenvector fields $X$ and $Y$  of $g^{(K)}$. 
The fact that this linear combination depends on the value of $Q$ according to \eqref{eq:hypersurfaceorthogonalKVF.1} is our main motivation to employ the more general $Q_\text{const}$-areal gauge as opposed to the $Q_0$-areal gauge which we use in \cite{abio2020}.

We also note that, similar to the standard Kasner solutions in \cite{abio2020}, the $\Lambda$-Kasner solutions discussed here have six natural isometries that preserve the form of the metric and map between different values of $K$. 
Each of these six isometries is defined as a map that swaps a pair of spatial coordinates. 
While each of these isometries can be applied to transform $g^{(K)}$ into an equivalent form with a different $K$-value preserving the areal gauge, none of these isometries simultaneously preserve the $Q_{\text{const}}$-areal gauge of a generic metric perturbation $g$. 
In particular, the isometry that interchanges the $x$ and $y$ coordinates --- i.e., the isometry corresponding to \eqref{eq:coordtrafo}-\eqref{eq:coordtrafo2} with $a_{11}=a_{22}=0$ and $a_{12}=a_{21}=1$ --- maps according to \eqref{eq:Qtransform} the constant value of $Q$ to a generally non-constant function form for $Q$.
Similar to the situation in \cite{abio2020}, we conclude from this that none of these isometries can be exploited to enlarge the range of allowed $K$-values for which Theorem~\ref{thm.main_theorem} applies. 
It follows that the restrictions \eqref{eq.JXnonzero.1}-\eqref{eq.JXzero_Lambdazero.1} for $K$ are geometric.

Finally, we remark that our theorem also extends the results on the extent of areal time in the strictly non-polarised case with $\Lambda>0$ in \cite{Smulevici:2011} to the \emph{polarised} $\Tbb^2$-symmetric setting with arbitrary $\Lambda\in\Rbb$, at least for cosmological solutions close to our subfamily of $\Lambda$-Kasner solutions (c.f. the discussion in Section~\ref{s.T2spacetimes.areal_gauge}).

\section{Proof of Theorem~\ref{thm.main_theorem}}
\label{s.main_proof}

\subsection{The polarised $\Tbb^2$-symmetric 
  Einstein equations as a first order symmetric hyperbolic Fuchsian system}
The first step of the proof of Theorem~\ref{thm.main_theorem} is to express the evolution system as a first order symmetric hyperbolic Fuchsian system. To this end, we set
\begin{equation}
\label{first_order_variables}
   (z_0, z_1, z_2) = (u, \del{t}u, \del{\theta}u), 
   \quad 
   \zeta = \del{\theta}\alpha,
 \end{equation}
 and then define new variables 
\begin{equation}
  \label{Udef}
  U = (w_0, w_1,w_2,\psi,\eta,\xi)^{\tr}
\end{equation}
by
 \begin{align}
   \label{eq:w0}
   z_0&=\ell\ln(t)-\ln(w_0),\\
   z_1&=\frac 1t (a+w_1),\\
   z_2&=\frac 1t w_2,\\
   \alpha&=1+\psi,\\
   \zeta&=\frac 1t\eta, \\
   \label{eq:defxi}
   \nu&=\ln(t)+\ln(\xi),
 \end{align}
 for some, so far, arbitrary real numbers $a$ and $\ell$. 
For our arguments below, we find it useful to define
\begin{align}
  \label{eq.bdef1}
  \bfr_P &:= P^2 (1 + \psi) \xi^2 , 
  \quad 
  \bfr_{J_X} := J_X^2 t^{2(1-2\ell)} w_0^4(1+\psi)\xi^2, \\ 
  \label{eq.bdef2}
  \bfr_\Lambda &:= 2 \Lambda t^{2(2-\ell)}w_0^2(1+\psi)\xi^2,
  \quad
  \bfr := \bfr_P + \bfr_{J_X} + 2 \bfr_\Lambda.
\end{align}

Using the above definitions, it is not difficult to verify via a straightforward calculation that the polarised $\Tbb^2$-symmetric vacuum Einstein equations \eqref{eq:vacevolFirst}-\eqref{eq:vacevolLast} are equivalent to the following evolution equations
\begin{align}
  \label{eq:mainevol1}
  B^0\del{t}U + B^1\del{\theta}U
    &=\frac{1}{t}\Bc\Pbb U + \frac{1}{t}F, \\
  \label{eq:mainevol2}
  \partial_t G
    &=\frac{P Q}{t}\sqrt{1+\psi}\,\xi^2
      -\frac{J_X}{t}\sqrt{1+\psi}\,t^{1-2\ell}\xi^2 w_0^4, \\
  \label{eq:mainevol3}
  \partial_t H
    &= -\frac{P}{t}\sqrt{1+\psi}\,\xi^2,
\end{align}
and constraints
 \begin{align}
   \label{eq:finalconstraint1}
   t{\del{\theta}\xi} &=2 w_2 (a+w_1) {\xi}-\frac{\eta {\xi}}{2(1+\psi) },\\
   \label{eq:finalconstraint2}
   {\del{\theta} w_0}&=-\frac 1t w_2 w_0,\\
   \label{eq:finalconstraint3}
   \del{\theta}\psi&=\frac 1t \eta,\\
   \label{eq:finalconstraint4}
   P {J_X} &=0,
 \end{align}
 where 
\begin{flalign}
  \Pbb&=
    \begin{pmatrix} 
     1 & 0 & 0 & 0 & 0 & 0\\
     0 & 0 & 0 & 0 & 0 & 0\\
     0 & 0 & 1 & 0 & 0 & 0\\
     0 & 0 & 0 & 0 & 0 & 0\\
     0 & 0 & 0 & 0 & 1 & 0\\
     0 & 0 & 0 & 0 & 0 & 1
    \end{pmatrix}, & \label{Pbbdef}\\
  B^0&=
    \begin{pmatrix} 
     1 & 0 & 0 & 0 & 0 & 0\\
     0 & 1 & 0 & 0 & 0 & 0\\
     0 & 0 & 1+\psi & 0 & 0 & 0\\
     0 & 0 & 0 & 1 & 0 & 0\\
     0 & 0 & 0 & 0 & 1 & 0\\
     0 & 0 & 0 & 0 & 0 & 4
    \end{pmatrix}, &\label{B0def}\\
  B^1&=
    \begin{pmatrix} 
     0 & 0 & 0 & 0 & 0 & 0\\
     0 & 0 & -(1+\psi) & 0 & 0 & 0\\
     0 & -(1+\psi) & 0 & 0 & 0 & 0\\
     0 & 0 & 0 & 0 & 0 & 0\\
     0 & 0 & 0 & 0 & 0 & 0\\
     0 & 0 & 0 & 0 & 0 & 0
    \end{pmatrix},&\label{B1def}\\
  \Bc&=
    \begin{pmatrix} 
     \ell - a - w_1  & 0 & 0 & 0 & 0 & 0\\
     0 & 1 & 0 & 0 & 0 & 0\\
     0 & 0 & 1+\psi & 0 & 0 & 0\\
     0 & 0 & 0 & 1 & 0 & 0\\
     0 & 0 & 4(1+\psi)(\bfr_{J_X} + \bfr_\Lambda - (a+w_1)\bfr) & 0 & 1 - \bfr & 0 \\
     0 & 0 & 0 & 0 & 0 &
     4 \bigl((a+w_1)^2-1+w_2^2(1+\psi)\bigr) + \bfr
    \end{pmatrix}, &\label{Bcdef}\\
  \intertext{and}
  F &= 
    \begin{pmatrix}
      0 \\
      \frac{1}{2}(w_2\eta + \bfr_{J_X} + \bfr_\Lambda-(a+w_1)\bfr)\\
      0\\
      -(1+\psi)\bfr\\
      0\\
      0
    \end{pmatrix}. &\label{Fdef}
\end{flalign}
The significance of this formulation of the polarised $\Tbb^2$-symmetric vacuum Einstein equations is that the main evolutionary part of the system --- that defined by \eqref{eq:mainevol1} and \eqref{Pbbdef}-\eqref{Fdef} --- is now in first order symmetric hyperbolic Fuchsian form; for details, see \cite{BOOS:2020} and the appendix of \cite{abio2020}. 
Theorem A.1 in \cite{abio2020} plays a central role in the proof below, in particular, in the proof of Proposition~\ref{prop:main_existence} below.

\subsection{Global existence of solutions of the initial value problem of the main evolution equations}
In this section, we establish a global existence result for the initial value problem consisting of system \eqref{eq:mainevol1} with coefficients defined by \eqref{Pbbdef}-\eqref{Fdef} for initial data $\mathring{U}$, i.e.,
\begin{equation}
  \label{eq:GIVP_CD}
  U(T_0)=\mathring{U},
\end{equation}
where at this point $T_0>0$ can be chosen arbitrarily. 
Proposition~\ref{prop:main_existence} provides global existence and first estimates of the behaviour of the fields at $t=0$, which are then exploited in the remainder of the proof of Theorem \ref{thm.main_theorem}. 
The proposition follows, as we show below, from an application of Theorem A.1 from \cite{abio2020}.

\begin{prop} 
  \label{prop:main_existence}
  Suppose $T_0>0$, $k \in \Zbb_{\geq 2}$, $Q, P \in \Rbb$, 
  and the constants $J_X$, $\Lambda,\ell,a\in \Rbb$ satisfy one
  of the following:
  \begin{align}
    \label{eq.JXnonzero}
    &\text{Case $1$ $\quad J_X \ne 0$,  $a<\ell \le \frac 12$ and   $a \in (-\infty, -1)$},\\
    \label{eq.JXzero_Lambdanonzero}
    &\text{Case $2$ $\quad  J_X=0$, $\Lambda \ne 0$, 
     $a<\ell \le 2$ and $a \in (-\infty, -1) \cup (1, 2)$},\\
    \label{eq.JXzero_Lambdazero}
    &\text{Case $3$ $\quad J_X=\Lambda = 0$, 
    $a<\ell$ and $a \in (-\infty, -1)\cup(1, \infty)$}.
  \end{align}
  Additionally, suppose
  \begin{equation*}
    \mathring{U}=\bigl(\mathring{w}_0,\mathring{w}_1,\mathring{w}_2,\mathring{\psi},\mathring{\eta},\mathring{\xi}\bigr)^{\tr}\in H^k(\Tbb,\Rbb^6)
  \end{equation*}
  is chosen so that $\mathring{\xi}>0$, and
set
\begin{equation}
\label{eq:definekappa0}
\kappa_0=\min\{1, a^2-1, \ell - a\}.
\end{equation}
Then, for every sufficiently small $R>0$ and $\sigma\in (0,\kappa_0)$, there exists a constant $R_0>0$ such that, if $\mathring{U}$ satisfies
  \begin{equation}
    \label{eq:CDSmallness}
    \norm{\mathring{U}}_{H^k}< R_0,
  \end{equation}
  there exists a unique solution 
  \begin{equation*}
    U \in C^0\bigl((0,T_0],H^k(\Tbb,\Rbb^6)\bigr)\cap  L^\infty\bigl((0,T_0],H^k(\Tbb,\Rbb^6)\bigr) \cap C^1\bigl((0,T_0],H^{k-1}(\Tbb,\Rbb^6)\bigr)\subset C^1\bigl((0,T_0]\times\Tbb,\Rbb^6\bigr)
  \end{equation*}
  of the GIVP\footnote{GIVP stands for \textit{global initial value problem}, which refers to solutions of symmetric hyperbolic Fuchsian initial value problems where the initial data is posed at some $T_0>0$ and the solution exists for all $t\in (0,T_0]$.} \eqref{eq:mainevol1} and \eqref{eq:GIVP_CD} such that
  \begin{equation}
    \label{eq.Ubnd}
    \norm{U}_{L^\infty((0,T_0]\times \Tbb)} < R,
  \end{equation}
 and the limit $\lim_{t\searrow 0} \Pbb^\perp U(t)$, denoted 
  \begin{equation}
    \label{eq.PbbperpU(0)}
    \Pbb^\perp U(0) = (0, \wt_1, 0, \psit, 0,0), 
  \end{equation}
  exists in $H^{k-1}(\Tbb,\Rbb^6)$.
  Finally, for $0<t<T_0$,  the solution $U$ satisfies  the energy estimate
  \begin{equation}  
    \label{eq:energyestimates}
    \norm{U(t)}_{H^k}^2+ \int_t^{T_0} \frac{1}{\tau} \norm{\Pbb U(\tau)}_{H^k}^2\, d\tau   \lesssim \norm{\mathring{U}}_{H^k}^2,
  \end{equation}
  the decay estimates    
  \begin{align}
    \label{eq:decayestimatesGIVP}
    \norm{\Pbb U(t)}_{H^{k-1}} 
      \lesssim t^{\kappa_0-\sigma}
    \AND 
    \norm{\Pbb^\perp U(t) - \Pbb^\perp U(0)}_{H^{k-1}} 
    \lesssim t + t^{2\kappa_0-2\sigma},
  \end{align}
  and the bound
  \begin{equation}
    \label{eq:wt1est}
    \norm{\wt_1}_{H^{k-1}}+\norm{\psit}_{H^{k-1}}\lesssim \norm{\mathring{U}}_{H^k}.
  \end{equation}
\end{prop}

Observe that, in contrast to Theorem~\ref{s.main_result}, we do not require $T_0$ to be small here. 
The smallness requirement for $T_0$ is introduced only in Section~\ref{s.proof_main_result} below for the reason explained there.

\begin{rem} 
  In Section ~\ref{s.proof_main_result}, we interpret the solutions specified in Proposition~\ref{prop:main_existence} as perturbations of $\Lambda$-Kasner solutions where the constant $a$ determines the Kasner parameter $K$ and the constants $Q$, $P$ and $J_X$ characterise the class of perturbations. 
  Note also that Proposition~\ref{prop:main_existence} only establishes the existence of solutions to \eqref{eq:mainevol1}; the remaining equations \eqref{eq:mainevol2}-\eqref{eq:finalconstraint4} are addressed in Propositions~\ref{prop:hoexp} and~\ref{prop:hoexp} below. 
  It is also worth pointing out that the requirement that $\xi>0$ --- see \eqref{eq:defxi} --- is fulfilled by the initial data choice $\mathring\xi>0$ since this condition is preserved by the last equation of the system \eqref{eq:mainevol1}. 
  Finally, we observe that Case 3, where $\bfr_{J_{X}} = \bfr_\Lambda = 0$ and $\bfr = \bfr_P = J_Y^2 (1+\psi)\xi^2$ (cf.\ \eqref{eq.bdef1}-\eqref{eq.bdef2}), is discussed in \cite{abio2020}. 
  Since $\ell$ can be taken to be an arbitrarily large value greater than $a$ in this case, we find that $\kappa_0=\min\{1, a^2-1\}$ which agrees with the result from \cite{abio2020}.
\end{rem}
\begin{proof}[Proof of Proposition~\ref{prop:main_existence}] 
The proof of this proposition follows from an application of Theorem A.1 from Appendix~A.2 of \cite{abio2020}. In order to apply this theorem, we first need to verify that the symmetric hyperbolic Fuchsian system \eqref{eq:mainevol1} satisfies the coefficient assumptions (i)-(vi) from Appendix~A.1 of \cite{abio2020}. 

To start, we observe that properties (i) and (ii) are obviously satisfied; that is, $U$ is an $\Rbb^N$-valued map, and the projection operator $\Pbb$ --- see \eqref{Pbbdef} --- satisfies $\Pbb^2 = \Pbb$, $\Pbb^{\tr} = \Pbb$, and $\del{t}\Pbb =0$. 
Regarding property (iii), it is clear from \eqref{B0def} and \eqref{Bcdef} that both $B^0$ and $\Bc$ depend smoothly on the components of $U$ and only $\Bc$ contains explicit dependence on $t$ through terms involving $\bfr_{J_X}, \bfr_\Lambda$, and hence also $\bfr$ (cf.\ \eqref{eq.bdef1}-\eqref{eq.bdef2}). Continuity of $\Bc$ at $t=0$ follows in each of the three cases \eqref{eq.JXnonzero}-\eqref{eq.JXzero_Lambdazero} from the assumptions on $\ell$. 
Furthermore, since, in each of the three cases \eqref{eq.JXnonzero} -\eqref{eq.JXzero_Lambdazero} we have $\ell > a, a^2>1$, there exist constants $R, \gamma_1, \gamma_2, \kappa > 0$ such that 
  \begin{equation*}
    0 < \kappa < \kappa_0=\min\{1, a^2-1, \ell - a\} 
    \AND 
   \gamma_1 > \frac{1}{1-R}      
  \end{equation*} 
  and the inequalities
  \[B^0 \le \frac 1\kappa \Bc \le \gamma_2 \id 
  \AND
  \frac{1}{\gamma_1} \id \le B^0
  \]
  hold for all $U \in \Rbb^6$ with $|U| < R$. 
  The remaining structural conditions in property (iii) are easily verified. 

  Next, we observe that property (iv) is satisfied by the source term $F$ since it is clear from \eqref{Fdef} and  the assumptions on $\ell$ in each of the three cases \eqref{eq.JXnonzero} -\eqref{eq.JXzero_Lambdazero} that $\Pbb F=0$,
  \[\Pbb^\perp F=\Ordc\Bigl(\frac{\lambda}R \Pbb U\otimes\Pbb U\Bigr),\]
  for some $\lambda=\Ord(R)$, and $F$ is continuous in $t$ including at $t=0$. Also, by \eqref{B0def} and \eqref{B1def}, we note that the matrix valued maps $B^0$ and $B^1$ satisfy the regularity and symmetry assumptions in property (v). Moreover, regarding property (vi), it is not difficult to verify that there exist positive constants $\theta, \beta$ such that
  \begin{equation}
    \Div\! B=
    \begin{pmatrix}
      0 & 0 & 0 & 0 & 0 & 0 \\
      0 & 0 & -\del\theta\psi & 0 & 0 & 0 \\
      0 & -\del\theta\psi & - \frac{\bfr}{t}(1+\psi) & 0 & 0 & 0 \\
      0 & 0 & 0 & 0 & 0 & 0 \\
      0 & 0 & 0 & 0 & 0 & 0 \\
      0 & 0 & 0 & 0 & 0 & 0 \\
    \end{pmatrix},
  \end{equation}
  satisfies $\Div\! B= \Ordc(\theta + \frac{\beta}{t}\Pbb U\otimes\Pbb U)$.

With the coefficient conditions (i)-(vi) from Appendix~A.1 of \cite{abio2020} verified, the proof, in particular the existence statement and the estimates \eqref{eq:energyestimates}-\eqref{eq:decayestimatesGIVP}, follows from an application of \cite[Theorem~A.1]{abio2020}, where we note that the constants $\kappa, \lambda, \text{ and } R$ can be chosen to satisfy $\kappa > \gamma_1 \lambda$. 
Finally, we conclude the proof by noting that the estimate \eqref{eq:wt1est} is a direct consequence of \eqref{eq:energyestimates} and \eqref{Pbbdef}.
\end{proof}

\subsection{Improved decay estimates for the solutions approaching the singularity}
\label{ss.improved_asymptotics}
Having established, for sufficiently small initial data $\norm{\mathring{U}}_{H^k}$, the existence of solutions to the GIVP consisting of \eqref{eq:mainevol1} and \eqref{eq:GIVP_CD}, the next step is to derive improved asymptotic estimates as $t \searrow 0$. 
This is carried out in the proof of the following proposition.

\begin{prop} \label{prop:hoexp}
Suppose $T_0>0$, $k \in \Zbb_{\geq 3}$, $\sigma\in (0,2\kappa_0/3)$, $Q, P \in \Rbb$, and suppose the constants $J_X,\Lambda,\ell,a\in \Rbb^3$ satisfy the conditions corresponding to one of the three cases \eqref{eq.JXnonzero}-\eqref{eq.JXzero_Lambdazero}.
Let 
\begin{equation}
\label{eq:HOCD}
\mathring{U}=\bigl(\mathring{w}_0,\mathring{w}_1,\mathring{w}_2,\mathring{\psi},\mathring{\xi},\mathring{\eta}\bigr)^{\tr}\in H^k(\Tbb,\Rbb^6),
\end{equation}
with $\mathring{w}_0>0$ and $\mathring\xi>0$, be chosen to
satisfy $\norm{\mathring{U}}_{H^k}< R_0$
for $R_0>0$ small enough so that it follows from Proposition~\ref{prop:main_existence} that there exists a unique solution
\begin{equation*}
U=\bigl(w_0,w_1,w_2,\psi,\xi,\eta\bigr)^{\tr} \in C^0\bigl((0,T_0],H^k(\Tbb,\Rbb^6)\bigr)\cap L^\infty\bigl((0,T_0],H^k(\Tbb,\Rbb^6)\bigr)\cap C^1\bigl((0,T_0],H^{k-1}(\Tbb,\Rbb^6)\bigr)
\end{equation*}
to the GIVP \eqref{eq:mainevol1} and \eqref{eq:GIVP_CD}.  
Let $(0, \wt_1, 0, \psit, 0,0)$ denote the limit $\lim_{t\searrow 0} \Pbb^\perp U(t)$ in $H^{k-1}(\Tbb,\Rbb^6)$.
Then there exist functions $\wt_2,\etat\in H^{k-2}(\Tbb)$ and $\ut, \nut\in H^{k-1}(\Tbb)$ such that
\begin{align}  
  \norm{t^{-1}w_2(t)-\ln(t)\del{\theta}\wt_{1}-\wt_2}_{H^{k-2}} &\lesssim
  t+t^{2\kappa_0-3\sigma}, \label{prop:hoexp.w2} \\
  \norm{t^{-1}\eta-\tilde{\eta}}_{H^{k-2}} &\lesssim
  t+t^{2\kappa_0-3\sigma}, \label{prop:hoexp.eta}
  \intertext{and}
\norm{\ln(w_0(t))+(a+\wt_1-\ell)\ln(t)-\ut}_{H^{k-1}}
&\lesssim t+t^{2\kappa_0-2\sigma}, \label{prop:hoexp.w0} \\
\norm{\ln(\xi)+(1-(a+\wt_1)^2)\ln(t)-\nut}_{H^{k-1}}
&\lesssim t+t^{2\kappa_0-2\sigma}, \label{prop:hoexp.xi}
\end{align}
for $0<t\le T_0$, where $\kappa_0$ is given by \eqref{eq:definekappa0}.
Furthermore, if the initial data \eqref{eq:HOCD} is chosen so that the constraint \eqref{eq:finalconstraint1} is satisfied at $t=T_0$, then 
the functions $\wt_1,\wt_2, \psit, \nut,\tilde{\eta}$ satisfy the asymptotic constraint equation
\begin{equation}
  \label{eq.AsymptoticConstraint}
  \partial_\theta \nut
        - 2 (a+\wt_1) \wt_2
        + \frac 12 (1+\psit)^{-1} \tilde\eta=0.
\end{equation}
\end{prop}
\begin{proof}
  The proof of this proposition follows closely the proof of Proposition~6.2 from  \cite{abio2020}. We start by writing 
  \eqref{eq:mainevol1} in the form
  \begin{equation}
      \label{eq.prop2.mainevol1}
      \del{t}U=\frac{1}{t}\Ac \Pbb U
      +\frac{1}{t}\bigl((B^0)^{-1}\Bc-\Ac\bigr)\Pbb U + \frac{1}{t}(B^0)^{-1}F-(B^0)^{-1}B^1\del{\theta}U,
  \end{equation}
  where
  \begin{equation} \label{Acdef}
      \Ac := \bigl((B^0)^{-1}\Bc\bigr)|_{U=0}=
      \begin{pmatrix} 
      \ell - a & 0 & 0 & 0 & 0&0\\
      0 & 1 & 0 & 0 & 0&0\\
      0 & 0 & 1 & 0 & 0&0\\
      0 & 0 & 0 & 1 & 0&0\\
      0 & 0 & 0 & 0 & 1&0\\
      0 & 0 & 0 & 0& 0 &a^2-1
    \end{pmatrix}.
  \end{equation}
  Defining
  \begin{equation} \label{Vdef}
    V= (V_1,V_2V_3,V_4,V_5,V_6)^{\tr} :=  e^{-\ln(t)\Ac\Pbb}U,
  \end{equation}
  where
  \begin{equation} \label{exp-mat}
      e^{-\ln(t)\Ac\Pbb} =  
      \begin{pmatrix} 
        \frac{1}{t^{\ell-a}} & 0 & 0 & 0 & 0 & 0\\
        0 & 1 & 0 & 0 & 0 & 0\\
        0 & 0 & \frac{1}{t} & 0 & 0 & 0\\
        0 & 0 & 0 & 1 & 0 & 0\\
        0 & 0 & 0 & 0 & \frac{1}{t} & 0\\
        0 & 0 & 0 & 0 & 0 & \frac{1}{t^{a^2-1}}
      \end{pmatrix},
  \end{equation}
  a short calculation shows that $V$ satisfies
  \begin{equation} 
    \label{eq.prop2.mainevol2}
    \del{t}V=
    \frac{1}{t}e^{-\ln(t)\Ac\Pbb}\bigl((B^0)^{-1}\Bc-\Ac\bigr)\Pbb U + \frac{1}{t} e^{-\ln(t)\Ac\Pbb}(B^0)^{-1}F- e^{-\ln(t)\Ac\Pbb}(B^0)^{-1}B^1\del{\theta}U.
  \end{equation}
  We then observe that this equation can be written in the form
  \begin{equation}
    \partial_t V    
    =
    \frac 1t\begin{pmatrix}
      -w_1 V_1\\
      \frac{1}{2}(w_2\eta + \bfr_{J_X} + \bfr_\Lambda-(a+w_1)\bfr)+t(1+\psi)\partial_\theta w_2\\
      \partial_\theta w_1\\
      -(1+\psi) \bfr\\
      -\bfr V_5+4(1+\psi)(\bfr_{J_X} + \bfr_\Lambda - (a+w_1)\bfr)V_3\\
      (\bigl(2aw_1+w_1^2+w_2^2(1+\psi)\bigr) + \frac{\bfr}4) V_6
    \end{pmatrix},
  \end{equation}
  or equivalently as
  \begin{equation} 
    \label{eq.prop2.mainevol3}
    \del{t}V = 
    \frac{1}{t}\begin{pmatrix}0\\0 \\ \del{\theta}\wt_1\\0\\0\\0\end{pmatrix}
    + \frac{1}{t}
    \begin{pmatrix} 
      (\wt_1-w_1)V_1 - \wt_1V_1\\0 \\ 0\\0\\h V_3 -\bfr V_5\\(2a\wt_1+\wt_1^2 + f + \frac \bfr4)V_6 
    \end{pmatrix} 
    + \frac{1}{t}
    \begin{pmatrix}
0\\
      \frac{1}{2}(\eta w_2 + \bfr_{J_X} + \bfr_\Lambda - (a+w_1)\bfr) +t(1+\psi)\del{\theta}w_2 \\
      \del{\theta}(w_1-\wt_1)\\
      -(1+\psi)\bfr \\ 0 \\ 0
    \end{pmatrix}, 
  \end{equation}
  where
  \begin{align}
    \label{eq.fdef}
    f=2aw_1+w_1^2-(2a \wt_1+\wt_1^2)+(1+\psi)w_2^2, \\
    \intertext{and}  
    \label{eq.hdef}
    h= 4((a+w_1)\bfr - \bfr_{J_X} - \bfr_\Lambda)(1+\psi).
  \end{align}
In order to integrate the equation, we introduce a third variable 
  \begin{equation} 
    \label{eq.Wdef}
    W:= 
    \begin{pmatrix} 
      t^{\wt_1}V_1 \\ V_2\\V_3-\ln(t)\del{\theta}\wt_1 \\
      V_4\\V_5 \\ t^{-(2a \wt_1+\wt_1^2)}V_6
    \end{pmatrix}
    = 
    \begin{pmatrix} 
      t^{a+\wt_1 - \ell}w_0 \\ w_1 \\ t^{-1}w_2-\ln(t)\del{\theta}\wt_1 \\
      \psi \\ t^{-1}\eta \\ t^{1-(a+\wt_1)^2}\xi
    \end{pmatrix},
  \end{equation}
  in terms of which we can express \eqref{eq.prop2.mainevol3} as
  \begin{equation} \label{W-evolve}
      \del{t}W=\Cc W + \Fc
  \end{equation}
  where
  \begin{align}
    \Cc &= \frac{1}{t}
    \begin{pmatrix} 
      \wt_1-w_1 & 0 & 0 & 0 & 0 & 0\\
      0 & 0 & 0 & 0 & 0 & 0\\
      0 & 0 & 0 & 0 & 0 & 0\\
      0 & 0 & 0 & 0 & 0 & 0\\
      0 & 0 & h & 0 & -\bfr & 0\\
      0 & 0 & 0 & 0 & 0 & \frac 14 \bfr + f 
    \end{pmatrix}, \label{Ccdef}\\
      \intertext{and}    
      \Fc &= \frac{1}{t}
    \begin{pmatrix}
0\\
      \frac{1}{2}(\eta w_2 + \bfr_{J_X} + \bfr_\Lambda - (a+w_1)\bfr) +t(1+\psi)\del{\theta}w_2 \\
      \del{\theta}(w_1-\wt_1)\\
      -(1+\psi)\bfr \\ -\ln(t)h\del{\theta}\wt_1 \\ 0
    \end{pmatrix}. \label{Fcdef}
  \end{align}  

  Integrating \eqref{W-evolve} in time yields
\begin{equation} \label{W-int}
    W(t)=W(t_0) + \int_{t_0}^t \Cc(\tau, W(\tau))W(\tau)+\Fc(\tau,W(\tau),\partial_\theta W(\tau))\,d\tau
\end{equation}
for $0<t\le t_0\le T_0$.
By the triangle inequality, and the Sobolev and product estimates --- see Proposition 2.4 and 3.7 from Chapter 13 of \cite{TaylorIII:1996} ---
we find, since $k-2\geq 1>1/2$, that
\begin{equation*}
    \norm{W(t)}_{H^{k-2}} \leq \norm{W(T_0)}_{H^{k-2}}+\int^{T_0}_t 
    \norm{\Cc(\tau, W(\tau))}_{H^{k-2}}
    \norm{W(\tau)}_{H^{k-2}}+\norm{\Fc(\tau, W(\tau),\partial_\theta W(\tau))}_{H^{k-2}}\,d\tau.
\end{equation*}
From this, we conclude via an application of Gr\"onwall's inequality that
\begin{equation} \label{W-bnd-A}
\norm{W(t)}_{H^{k-2}} \leq e^{\int^{T_0}_t 
    \norm{\Cc(\tau, W(\tau))}_{H^{k-2}}\,d\tau}
\biggl(\norm{W(T_0)}_{H^{k-2}}+\int^{T_0}_t\norm{\Fc(\tau, W(\tau),\partial_\theta W(\tau))}_{H^{k-2}}\,d\tau\biggr).
\end{equation}
As a consequence of \eqref{eq.fdef}, \eqref{eq.hdef}, \eqref{eq.bdef2}, \eqref{Ccdef}, and \eqref{Fcdef}, we observe, with the help of the energy and decay estimates \eqref{eq:energyestimates}-\eqref{eq:decayestimatesGIVP} and the Sobolev and product estimates --- see Proposition 2.4 and 3.7 from Chapter 13 of \cite{TaylorIII:1996} --- that
\begin{equation} \label{W-bnd-B}
    \int_{t_0}^t\norm{\Cc(\tau, W(\tau))}_{H^{k-2}}+ \norm{\Fc(\tau, W(\tau),\partial_\theta W(\tau))}_{H^{k-2}}\,d\tau \lesssim \bigl(t+t^{2\kappa_0-3\sigma}\bigr)-\bigl(t_0+t_0^{2\kappa_0-3\sigma}\bigr).
\end{equation}
Thus, by \eqref{W-bnd-A}, we have 
\begin{equation*}
   \sup_{0<t<T_0} \norm{W(t)}_{H^{k-2}} \lesssim 1.
\end{equation*}
With the help of this uniform bound, we deduce
from \eqref{W-int} and another application of the product, Sobolev, and triangle inequalities, that
\begin{align*}
\norm{W(t)-W(t_0)}_{H^{k-2}} 
&\leq   \int_{t_0}^t\norm{\Cc(\tau, W(\tau))}_{H^{k-2}}
\norm{W(\tau)}_{H^{k-2}}\,d\tau + \int_{t_0}^t\norm{\Fc(\tau, W(\tau),\partial_\theta W(\tau))}_{H^{k-2}}
\,d\tau \\
&\lesssim  \int_{t_0}^t\norm{\Cc(\tau, W(\tau))}_{H^{k-2}}+ \norm{\Fc(\tau, W(\tau),\partial_\theta W(\tau))}_{H^{k-2}}\,d\tau.
\end{align*}
From this inequality and \eqref{W-bnd-B}, we conclude that the limit $\lim_{t\searrow 0} W(t)$ converges to an element of $H^{k-2}(\Tbb)$, and denoting this element by $W(0)$, we further conclude that $W(t)$ can be extended to a uniformly continuous map
\begin{equation*}
    W\in C^0\bigl([0,T_0],H^{k-2}(\Tbb,\Rbb^6)\bigr)
\end{equation*}
that satisfies
\begin{equation} \label{W-bnd-C}
\norm{W(t)-W(0)}_{H^{k-2}} 
\lesssim  t+t^{2\kappa_0-3\sigma}
\end{equation}
for $0<t\le T_0$. 
It is then clear that the estimates \eqref{prop:hoexp.w2}-\eqref{prop:hoexp.eta} follow directly from \eqref{eq.Wdef} and \eqref{W-bnd-C}.

With the main estimate for $W$ complete, we now turn our attention to establishing improved estimates for the first and last components of $W$. We begin by considering the first component, that is, $W_1=t^{a+\wt_1-\ell}w_0$. Since $\mathring w_0>0$, it follows that $w_0$ and hence $W_1$ are strictly positive everywhere these maps are defined. Consequently, we see by \eqref{W-evolve} that 
\[\partial_t \ln(W_1)=(\wt_1-w_1)/t.\]
Integrating this equation in time and invoking once more the Sobolev and product inequalities, we find, with the help of  \eqref{eq:decayestimatesGIVP}, that
\[\norm{\ln(W_1(t))-\ln(W_1(\tilde t))}_{H^{k-1}}\lesssim \int_{\tilde t}^t \bigl(s + s^{2\kappa_0-2\sigma} \bigr) s^{-1}ds\lesssim (t+t^{2\kappa_0-2\sigma})-({\tilde t}+{\tilde t}^{2\kappa_0-2\sigma})\]
for any $\tilde{t}\in (0,T_0]$ and $t\in (\tilde{t},T_0]$.
So, given any sequence $(t_n)$ in $(0,T_0]$ approaching zero, we deduce that the sequence $(\ln(W_1(t_n)))$ converges in the $H^{k-1}$-norm to a limit, which we preliminarily call $\ln(W_1(0))$, and that
\begin{equation}
  \label{prop:hoexp.w0.pre}
  \norm{\ln(W_1(t))-\ln(W_1(0))}_{H^{k-1}}\lesssim t+t^{2\kappa_0-2\sigma}.
\end{equation}
Since it follows from \eqref{eq.Wdef} together with \eqref{first_order_variables} and \eqref{eq:w0}  that 
\[\ln(W_1(t))=\ln(w_0(t))+(a+\wt_1-\ell)\ln(t),\]
we obtain \eqref{prop:hoexp.w0} from \eqref{prop:hoexp.w0.pre} if we rename $\ln(W_1(0))$ by $\ut$.

In the same way we can improve the estimate for $W_6=t^{1-(a+\wt_1)^2}\xi$, the last component of $W$. Since $\mathring \xi>0$, it follows that $\xi$ and hence $W_6$ are strictly positive everywhere these maps are defined. By \eqref{W-evolve}, we see that $w_6$ satisfies
\begin{equation}\label{w6-evolve}
\partial_t \ln(W_6)=(\frac 14\bfr-f)/t.
\end{equation}
Essentially the same arguments used for the component $W_1$ can be applied to establish that given any sequence $(t_n)$ in $(0,T_0]$ approaching zero, the sequence $(\ln(W_6(t_n)))$ converges in the $H^{k-1}$-norm to a limit, which we preliminarily call $\ln(W_6(0))$, and that
\begin{equation}
  \label{prop:hoexp.xi.pre}
  \norm{\ln(W_6(t))-\ln(W_6(0))}_{H^{k-1}}\lesssim t+t^{2\kappa_0-2\sigma}.
\end{equation}
However, it follows from \eqref{first_order_variables}, \eqref{eq:defxi}, and \eqref{eq.Wdef} that we have 
\[\ln(W_6)=\ln(\xi)+(1-(a+\wt_1)^2)\ln(t)=(a+\wt_1)^2\ln(t)+\nu.\]
We then obtain \eqref{prop:hoexp.xi} from \eqref{prop:hoexp.xi.pre} if we rename $\ln(W_6(0))$ by $\nut$.

To complete the proof, we note that the asymptotic constraint \eqref{eq.AsymptoticConstraint} can be established using the same  argument as in the proof of Prop.~6.2 in \cite{abio2020} by studying the limit $t\searrow 0$ of \eqref{eq:finalconstraint1} using \eqref{eq:wt1est} and \eqref{prop:hoexp.w2}-\eqref{prop:hoexp.xi} together with the Sobolev and product estimates observing that $k\ge 3$.
\end{proof}

\subsection{Solutions to the full polarised $\Tbb^2$-symmetric  Einstein equations}

Propositions \ref{prop:main_existence} and \ref{prop:hoexp} establish the global existence of solutions and detailed asymptotic estimates for solutions of the main evolution system, consisting of equations \eqref{eq:mainevol1} and \eqref{Pbbdef}--\eqref{Fdef}. 
In addition, the propagation of the constraint \eqref{eq:finalconstraint1} follows from arguments much like those in \cite{isenberg1999}, and its asymptotic form in our variables is established in Proposition~\ref{prop:hoexp}.
To obtain solutions of the full polarised $\Tbb^2$-symmetric Einstein equations, for initial data sufficiently close to $\Lambda$-Kasner initial data, we now impose \eqref{eq:mainevol2}, \eqref{eq:mainevol3}, \eqref{eq:finalconstraint2}, and \eqref{eq:finalconstraint3}.

\begin{prop}
  \label{prop:fullEFE}
  Consider the same conditions as specified in the hypothesis for Proposition~\ref{prop:hoexp}, and
  let $U$ be the solution to the GIVP \eqref{eq:mainevol1} and \eqref{eq:GIVP_CD}  determined by initial data $\mathring{U}$ with $\norm{\mathring{U}}_{H^k}< R_0$ and $\mathring{\xi}>0$ and $\mathring{w_0}>0$.
  Then for given, not necessarily small, $\mathring{G}, \mathring{H}\in H^{k}(\Tbb)$, the Cauchy problem consisting of evolution equations \eqref{eq:mainevol2}-\eqref{eq:mainevol3} and the initial condition $(G,H)|_{t=T_0}=(\mathring{G},\mathring{H})$ has a unique solution
  \[ (G, H) \in C^1\bigl((0,T_0],H^k(\Tbb,\Rbb^2)\bigr)\cap L^\infty\bigl((0,T_0],H^k(\Tbb,\Rbb^2)\bigr),\]
  and there exist functions $\Gt,\Ht\in H^{k-1}(\Tbb)$ such that
  \begin{align}
    \label{eq:GLimit}
    \Bnorm{G(t)-\Gt}_{H^{k-1}}\lesssim&|Q||P|t^{2(\min_{\theta\in(0,2\pi]}\{(a+\wt_1)^2\}-1)-\sigma}
    +|J_X|t^{2(\min_{\theta\in(0,2\pi]}\{(a+\wt_1)^2-2\wt_1\}-1+2\ell-2a)-\sigma},\\
    \label{eq:HLimit}
    \Bnorm{H(t)-\Ht}_{H^{k-1}}\lesssim&|P|t^{2(\min_{\theta\in(0,2\pi]}\{(a+\wt_1)^2\}-1)- \sigma},
  \end{align}  
  for all $t\in (0,T_0]$.
  Moreover, if the initial data is chosen to also satisfy constraints \eqref{eq:finalconstraint2}-\eqref{eq:finalconstraint3} at $t=T_0$, then the functions $\ut$, $\wt_2$, $\psit$ and $\etat$ satisfy the asymptotic constraint equations
\begin{align}
  \label{eq:AsymptConstr2}
  \partial_\theta\ut+\wt_2=0, \\
  \label{eq:AsymptConstr3}
  \del{\theta}\psit=\etat.
\end{align}
\end{prop}

\begin{rem}
  The constants $Q,P,J_X$ have been included in the estimates \eqref{eq:GLimit} and \eqref{eq:HLimit} in order to clarify which terms are relevant in each of the Cases 1 to 3 (cf. \eqref{eq.JXnonzero}-\eqref{eq.JXzero_Lambdazero}).
  In particular, in Cases 2 and 3, we note that $J_X=0$ and the second term in \eqref{eq:GLimit} vanishes.
\end{rem}

\begin{proof} 
Integrating \eqref{eq:mainevol1}-\eqref{eq:mainevol2}, we find 
\begin{align*}
    &G(t) = \mathring{G} 
     + QP  \int_{T_0}^t \sqrt{1 + \psi}\, \underbrace{e^{2\nut} \exp\bigl(2(\ln(\xi)+(1-(a+\wt_1)^2)\ln(s)-\nut)\bigr)
      s^{2((a+\wt_1)^2-1)}}_{=\xi^2(s)} s^{-1}\rmd s \\ \nonumber
    & - J_X \int_{T_0}^t \Bigl(\sqrt{1 + \psi}\, s^{1-2\ell}
      \underbrace{e^{2\nut} 
      \exp\bigl(2(\ln(\xi)+(1-(a+\wt_1)^2)\ln(s)-\nut)\bigr)
      s^{2((a+\wt_1)^2-1)}}_{=\xi^2(s)}\\
      &\qquad\qquad\qquad
      \underbrace{e^{4\ut} 
      \exp\bigl({4(\ln(w_0(t))+(a+\wt_1-\ell)\ln(s)-\ut)}\bigr)
      s^{4(\ell-(a+\wt_1))}}_{=w_0^4(s)}s^{-1}\Bigr)      
      \rmd s, \\
    &H(t) = \mathring{H}
      - P\int_{T_0}^t \sqrt{1 + \psi}\, \underbrace{e^{2\nut} \exp\bigl(2(\ln(\xi)+(1-(a+\wt_1)^2)\ln(s)-\nut)\bigr)
      s^{2((a+\wt_1)^2-1)}}_{=\xi^2(s)} s^{-1} \rmd s,
\end{align*}
where $\mathring G$ and $\mathring H$ are the initial data for $G$ and $H$ at $t=T_0$. For arbitrary $t$ and $t_0$ with $0<t_0\le t\le T_0$, we therefore find
\begin{equation}
  \label{eq:Gestimate}
  \begin{split}
  \norm{G(t)-G(t_0)}_{H^{k-1}}\lesssim &|Q||P|\bigl(t^{2(\min_{\theta\in(0,2\pi]}\{(a+\wt_1)^2\}-1)-\sigma}-t_0^{2(\min_{\theta\in(0,2\pi]}\{(a+\wt_1)^2\}-1)-\sigma}\bigr)\\
  +&|J_X|\bigl(
    t^{2(\min_{\theta\in(0,2\pi]}\{(a+\wt_1)^2-2\wt_1\}-1+2\ell-2a)-\sigma}
    -t_0^{2(\min_{\theta\in(0,2\pi]}\{(a+\wt_1)^2-2\wt_1\}-1+2\ell-2a)-\sigma}
  \bigr)
\end{split}
\end{equation}
by using the Sobolev and product estimates as well as \eqref{prop:hoexp.w0} and \eqref{prop:hoexp.xi} in a similar fashion as in the proof of Proposition~6.3 from \cite{abio2020}. Now, we claim that our hypotheses guarantee that in each of the three cases \eqref{eq.JXnonzero}-\eqref{eq.JXzero_Lambdazero} each of the exponents of $t$ in \eqref{eq:Gestimate} is positive. To see why this is the case, we note that $a^2>1$ and $\ell-a>0$ in all three cases \eqref{eq.JXnonzero}-\eqref{eq.JXzero_Lambdazero} and  $a<-1$ if $J_X\not =0$ (Case 1, see \eqref{eq.JXnonzero}). Because of these inequalities, we can guarantee the positivity
of the exponents by choosing 
$R_0$ and therefore  $\wt_1$ (according to \eqref{eq:CDSmallness} and \eqref{eq:wt1est}) small enough to ensure that
$(a+\wt_1)^2 -1> 0$ and $(a+\wt_1)^2-2\wt_1 -1> 0$.
It then follows from the completeness of the Sobolev space $H^{k-1}$ that $G(t)$ converges to a limit $\Gt\in H^{k-1}$ at $t=0$ and we obtain \eqref{eq:GLimit} from \eqref{eq:Gestimate}. For $H$, similar arguments
can be used to show that
\begin{equation}
  \norm{H(t)-H(t_0)}_{H^{k-1}}\lesssim |P|\bigl(t^{2(\min_{\theta\in(0,2\pi]}\{(a+\wt_1)^2\}-1)- \sigma}-t_0^{2(\min_{\theta\in(0,2\pi]}\{(a+\wt_1)^2\}-1)- \sigma}\bigr),
\end{equation}
which then allows us to deduce the existence of a  limit $\Ht\in H^{k-1}$ at $t=0$ for which \eqref{eq:HLimit} holds.

Assuming now that the constraint \eqref{eq:finalconstraint2} is satisfied at $t = T_0$, it follows from the evolution system \eqref{eq:mainevol1} and \eqref{Pbbdef}-\eqref{Fdef} that the constraint remains satisfied for all $t\in (0, T_0]$. 
The evolution equation for $w_0$ obtained from \eqref{eq:mainevol1} implies that $w_0 > 0$ for as long as this quantity is defined if the initial data $\mathring{w_0}$ is positive. 
This allows us to divide \eqref{eq:finalconstraint2} by $w_0$ to obtain
\[\del{\theta}\ln(w_0) + t^{-1}w_2=0.\]
Since this equation is satisfied at all $t\in (0,T_0]$, we conclude from the triangle inequality 
\begin{align*}
  0 = & \Bnorm{\del{\theta}\ln(w_0) + t^{-1}w_2}_{H^{k-2}} \\
  =&\Bnorm{
      \del{\theta}(\ln(w_0) +(a+\wt_1-\ell)\ln(t)-\ut)
      +(t^{-1}w_2-\del{\theta}\wt_1\ln(t)-\wt_2)
      +(\del{\theta}\ut +\wt_2)
    }_{H^{k-2}}\\
  \ge & -\Bnorm{\ln(w_0) +(a+\wt_1-\ell)\ln(t)-\ut}_{H^{k-1}}
      -\Bnorm{ (t^{-1}w_2 -\ln(t)\del{\theta}\wt_{1}-\wt_2)}_{H^{k-2}}
      +\Bnorm{\del{\theta}\ut+\wt_2}_{H^{k-2}}.
\end{align*}
But it follows from \eqref{prop:hoexp.w2} and \eqref{prop:hoexp.w0} that the first two terms approach zero as $t\searrow 0$, and consequently, letting $t\searrow 0$ in the above inequality yields \eqref{eq:AsymptConstr2}.
The asymptotic constraint \eqref{eq:AsymptConstr3} can be established in a similar manner.
\end{proof}

\subsection{Completing the proof of the main result}
\label{s.proof_main_result}
We are now ready to complete the proof of Theorem~\ref{thm.main_theorem} by combining the results of the above subsections. We begin by choosing an integer $k\ge 3$ and constants $Q, \Lambda, J_X\in \Rbb$. We then select constants $P$ and $K$ according to \eqref{eq.JXnonzero.1}-\eqref{eq.JXzero_Lambdazero.1}  
and set $a=(1-K)/2$. It is important to note that, (i) the constants $\Lambda$ and $K$ determine the background $\Lambda$-Kasner solution $g^{(K)}$ to be perturbed, (ii) the constant $\kappa_0$, as defined by \eqref{eq:definekappa0}, agrees in each of the three cases with \eqref{eq:kappa0final}, and, (iii) that the constant $\ell$ is determined according to \eqref{eq:kappa0final}.

Fixing a $T_0>0$ which is small enough to make sense of the expansions \eqref{eq:KasnersolFirst}--\eqref{eq:KasnersolLast}, it follows from \eqref{first_order_variables}--\eqref{eq:defxi} that the initial data for the evolution system \eqref{eq:mainevol1}--\eqref{eq:mainevol3} at $t=T_0$ determined by the $\Lambda$-Kasner metric $g^{(K)}$ is of the form
  \begin{gather*}
    w_0^{(K)}(T_0)=T_0^{\ell-a}(1+\ldots),\quad
    w_1^{(K)}(T_0)=0+\ldots,\quad
    w_2^{(K)}(T_0)=0,\quad
    \psi^{(K)}(T_0)=0+\ldots,\\
    \eta^{(K)}(T_0)=0,\quad
    \xi^{(K)}(T_0)=T_0^{a^2-1}(1+\ldots),\quad
                    Q^{(K)}=G^{(K)}(T_0)=H^{(K)}(T_0)=0.
  \end{gather*}
  For a sufficiently small choice of $R_0>0$, this Kasner data satisfies the condition $\norm{\mathring{U}}_{H^k}< R_0$ (see \eqref{eq:CDSmallness}) of Proposition~\ref{prop:main_existence} provided $T_0$ is chosen sufficiently small since $a,\ell$ satisfy $a^2>1$ and $\ell-a>0$ by \eqref{eq.JXnonzero.1}-\eqref{eq:kappa0final}. 
  It follows from Proposition~\ref{prop:main_existence} that we can then interpret solutions to \eqref{eq:mainevol1}-\eqref{eq:mainevol3} that are generated from general initial data $\mathring{U}$ at $t=T_0$ that satisfies $\norm{\mathring{U}}_{H^k}< R_0$ for sufficiently small $T_0$ and sufficiently small $R_0>0$ (but not necessarily small $\norm{\mathring{G}}_{H^k}$, $\norm{\mathring{H}}_{H^k}$, $Q$, $P$ and $J_X$) as \emph{perturbed $\Lambda$-Kasner solutions} provided the initial data satisfies the constraints \eqref{eq:finalconstraint1}-\eqref{eq:finalconstraint4} and the inequalities $\mathring{w_0},\mathring{\xi}>0$.

Now, an initial data set $(\mathring u, \mathring v, \mathring\nu, \mathring\alpha,\mathring{G}, \mathring{H})$ for the original variables $(u,\nu,\alpha,G,H)$ that satisfies the constraints \eqref{eq:vacconstraint} and \eqref{eq:vacevolLast} with $\partial_t u(T_0)=\mathring v$, and the smallness condition \eqref{eq:PerturbedKasnerCauchyData}, determines, as a consequence  of \eqref{eq:CDSmallness} and the transformation \eqref{first_order_variables}-\eqref{eq:defxi}, an initial data set $\mathring{U}$ for the system \eqref{eq:mainevol1}-\eqref{eq:mainevol3} satisfying $\norm{\mathring{U}}_{H^k}< R_0$ at $t=T_0$, the constraints \eqref{eq:finalconstraint1}-\eqref{eq:finalconstraint4} and the inequalities $\mathring{w_0},\mathring{\xi}>0$.
Consequently, for $T_0,R_0$ chosen sufficiently small, Proposition~\ref{prop:main_existence} establishes the existence of a classical solution $(u,\nu,\alpha,G,H)$ to  \eqref{eq:vacevolFirst}-\eqref{eq:vacevolH} on $(0,T_0]$ 
with regularity determined by \eqref{eq:thmreg1} and \eqref{eq:thmreg2}. 
The limits $\wt_1$ and $\psit$ in $H^{k-1}(\Tbb)$ together with \eqref{eq:decayestimatesGIVP} imply that \eqref{eq:ualphaest} holds for $\Kt$ and $\alphat$ in $H^{k-1}(\Tbb)$ given by $\Kt=K-2\wt_1$ and $\alphat=1+\psit$. 
We further observe that \eqref{eq:smallnessKalpha} is a direct consequence of \eqref{eq:wt1est}.
Proposition~\ref{prop:hoexp} implies the existence of limits $\ut, \nut\in H^{k-1}(\Tbb)$ for which \eqref{eq:unuest} holds as a consequence of \eqref{prop:hoexp.w0} and \eqref{prop:hoexp.xi}.
We additionally note that, due to the asymptotic constraints \eqref{eq:AsymptConstr2}-\eqref{eq:AsymptConstr3} from Proposition~\ref{prop:fullEFE}, the estimates implied by \eqref{prop:hoexp.w2} and \eqref{prop:hoexp.eta} also follow from \eqref{eq:ualphaest} and are therefore redundant, and that the asymptotic constraint \eqref{eq.AsymptoticConstraint} takes the form \eqref{eq.AsymptoticConstraintFinal}.  
Finally, the estimates \eqref{eq:GLimit.Final} and \eqref{eq:HLimit.Final} follow from the estimates \eqref{eq:GLimit} and \eqref{eq:HLimit} from Proposition~\ref{prop:fullEFE}, where
we note that $\ell=1/2$ in the case $J_X\not=0$ (Case~1).

To complete the proof, we have computed the Kretschmann scalar $\mathcal{I}$ for a metric of the form \eqref{T2metric} using computer algebra \cite{xcoba,xact}. 
Using the resulting expression for $\mathcal{I}$, it is then not difficult to verify, with the help of
Propositions~\ref{prop:main_existence}, \ref{prop:hoexp} and \ref{prop:fullEFE}, the definitions \eqref{first_order_variables} and \eqref{eq:w0}-\eqref{eq:defxi}, and the definition of the limit functions $\ut, \nut, \alphat, \Kt$ given above, that $\mathcal{I}$ satisfies  \eqref{eq.KretschmannLeadingOrder}.

\appendix

\section{$\Lambda$-Kasner solutions}
\label{sec:Kasnersol}

Let $\tb$, $\xb_1$, $\xb_2$, $\xb_3$ be coordinates on $\Rbb^+ \times \Tbb^3$. It is straightforward to show that \eqref{eq:KasnermetricGauss}
is a spatially homogeneous family of solutions of \eqref{eq:EFE} for arbitrary $\Lambda\in\Rbb$, which we call \emph{$\Lambda$-Kasner solutions} \cite{Garfinkle2021}, provided: (i) the relations \eqref{eq:Kasner} hold for the \emph{Kasner exponents} $p_1,p_2,p_3\in\Rbb$, (ii) the functions $\LKone(z)$ and $\LKtwo(z)$ in \eqref{eq:KasnermetricGauss} are the analytic functions on the complex plane given by
\begin{equation}
  \label{eq:K1K2}
  \LKone(z)=\frac 14(e^z+e^{-z})^2=\cosh^2(z),
  \quad 
  \LKtwo(z)=\frac{e^z-e^{-z}}{z(e^z+e^{-z})}=\frac{\tanh{z}}z,
\end{equation}
respectively, and, (iii) $z=z(\tb)$ is given by \eqref{eq:defz} 
where $\sqrt{\cdot}$ is one of the two branches of the complex square root function. Observe here that we consider $\LKtwo(z)$ as the function analytically extended from $\Cbb\backslash\{0\}$ to $\Cbb$ and that the values of the functions $\LKone(z(\tb))$ and $\LKtwo(z(\tb))$ are the same for both branches of the complex root function.
Note also that
\[\LKone(0)=\LKtwo(0)=1,\]
and that both $\LKone(z)$ and $\LKtwo(z)$ are real-valued for all $z\in\Rbb$ and $z\in i\Rbb$.
Due to these relations, it evident that the $\Lambda$-Kasner solutions reduce to the standard vacuum \emph{Kasner solutions} \cite{abio2020,kasner1921} in the case $\Lambda=0$, and that for \emph{arbitrary} $\Lambda\neq 0$, the $\Lambda$-Kasner solutions approach the standard vacuum Kasner solutions asymptotically as $\tb\searrow0$ according to
\[g^{(K)}=-\rmd \tb^2+\sum_{i=1}^3 \tb^{2p_i}(1+\Ord(z(\tb)^2))^{2p_i}dx_i^2.\]

The first fundamental form induced on the $\tb=const$-surfaces is
\begin{equation}
  \label{eq:KasnermetricGaussFirst}
  \gamma^{(K)}=\LKone^{2/3}\bigl(z(\tb)\bigr)\sum_{i=1}^3 \tb^{2p_i}\LKtwo^{2p_i}\bigl(z(\tb)\bigr)d\xb_i^2,
\end{equation}
while the Weingarten map (i.e., the mixed-component second fundamental form) is
\begin{equation}
  \label{eq:KasnerWeingarten}
  \Kbb^{(K)}=-\frac 1{\tb}\sum_{i=1}^3 \left(\frac{p_i}{\LKone\bigl(z(\tb)\bigr)\LKtwo\bigl(z(\tb)\bigr)}+\frac{2 \LKtwo\bigl(z(\tb)\bigr) z^2(\tb)}{3}\right)\partial_{\xb_i}\otimes d\xb_i.
\end{equation}
The eigenvalues of $-\tb\Kbb^{(K)}$ therefore agree with the Kasner exponents $p_1$, $p_2$ and $p_3$ in the limit $\tb\searrow0$.
Except in non-generic cases in which two of these Kasner exponents agree,
each coordinate vector field $\partial_{\xb_1}$, $\partial_{\xb_2}$, $\partial_{\xb_3}$  spans  one of the distinct $1$-dimensional eigenspaces of $\Kbb^{(K)}$ for all sufficiently small $\tb>0$, and is therefore geometrically distinguished.

The $\Lambda$-Kasner metric \eqref{eq:KasnermetricGauss} is spatially homogeneous and therefore a member of the polarised $\Tbb^2$-symmetric class; see Section~\ref{s.T2spacetimes}. 
However, \eqref{eq:KasnermetricGauss} is clearly not expressed in an areal gauge \eqref{T2metric} coordinates. 
In order to bring this metric to this gauge representation, we consider the coordinate transformation
\begin{equation}
  \label{eq:KasnerCoordTrafo}
  (t,\theta,x,y)=\Bigl(\frac {\tb^{1-p_1}}{(1-p_1)^{1-p_1}} \LKone^{2/3}\bigl(z(\tb)\bigr) \LKtwo ^{1-p_1}\bigl(z(\tb)\bigr), (1-p_1)^{p_1}\xb_1, (1-p_1)^{p_2} \xb_3, (1-p_1)^{p_3} \xb_2\Bigr).
\end{equation}
The resulting metric is of the form \eqref{T2metric} given by the following partially implicit relations
\begin{align}
  e^{2u^{(K)}}&= t^{1-K} \LKone^{-2K/3}\bigl(z(\tb(t))\bigr),\\
  e^{2\nu^{(K)}}&= t^{(1-K)^2/2}\LKone^{-(K-3)(K+1)/3}\bigl(z(\tb(t))\bigr),\\
  \alpha^{(K)}&=\left[\frac{3}{3+(K^2+3) z^2(\tb(t)) \LKone\bigl(z(\tb(t))\bigr) \LKtwo^2\bigl(z(\tb(t))\bigr)}\right]^2,\\
          Q^{(K)}&=G^{(K)}=H^{(K)}=0,
\end{align}
where $K\in\Rbb$ determines the Kasner exponents via \eqref{eq:Kasnerexpo}.
It turns out that
\begin{align}
  \label{eq:Kasnerexp1}
  e^{2u^{(K)}}&= t^{1-K}\left(1-\frac{8 \Lambda  K t^{(K^2+3)/2}}{\left(K^2+3\right)^2}+\Ord\left(\Lambda^2\, t^{K^2+3}\right)\right),\\  
  e^{2\nu^{(K)}}&= t^{(1-K)^2/2}
    \left(
      1
      -\frac{4 \Lambda\left(K-1\right)\left(K+3\right)
        t^{(K^2+3)/2}}{\left(K^2+3\right)^2}
      +\Ord\left(\Lambda^2\, t^{K^2+3}\right)
                  \right),\\
  \label{eq:Kasnerexp3}
  \alpha^{(K)}&=1-\frac{8 \Lambda   t^{(K^2+3)/2}}{K^2+3}+\Ord\left(\Lambda^2\, t^{K^2+3}\right),
\end{align}
near $t=0$.
It is clear from equation \eqref{eq:KasnerCoordTrafo} that the $t=const$-surfaces and the $\tb=const$-surfaces yield the same foliation of Kasner spacetimes. The map $\Kbb^{(K)}$ in \eqref{eq:KasnerWeingarten} is therefore the Weingarten map of the $t=const$-surfaces, and, except for degenerate cases, each coordinate vector field $\partial_{\theta}$, $\partial_{x}$, $\partial_{y}$ spans one of the distinct $1$-dimensional eigenspaces. In the degenerate cases the Kasner parameter $K$ 
takes the special values $K=0$, $K=\pm 1$, or $K=\pm 3$.

\bibliographystyle{amsplain}
\bibliography{refs}

\providecommand{\bysame}{\leavevmode\hbox to3em{\hrulefill}\thinspace}
\providecommand{\MR}{\relax\ifhmode\unskip\space\fi MR }
\providecommand{\MRhref}[2]{%
  \href{http://www.ams.org/mathscinet-getitem?mr=#1}{#2}
}
\providecommand{\href}[2]{#2}
\begin{thebibliography}{10}

\bibitem{ames2013a}
E.~Ames, F.~Beyer, J.~Isenberg, and P.~G. LeFloch, \emph{Quasilinear
  {{Hyperbolic Fuchsian Systems}} and {{AVTD Behavior}} in {$T^2$}-{{Symmetric
  Vacuum Spacetimes}}}, Ann. Henri Poincar{\'e} \textbf{14} (2013), no.~6,
  1445--1523,
  DOI:~\href{https://doi.org/10.1007/s00023-012-0228-2}{10.1007/s00023-012-0228-2}.

\bibitem{ames2017}
\bysame, \emph{A class of solutions to the {{Einstein}} equations with {{AVTD}}
  behavior in generalized wave gauges}, J. Geom. Phys. \textbf{121} (2017),
  42--71,
  DOI:~\href{https://doi.org/10.1016/j.geomphys.2017.06.005}{10.1016/j.geomphys.2017.06.005}.

\bibitem{abio2020}
E.~Ames, F.~Beyer, J.~Isenberg, and T.~A. Oliynyk, \emph{{Stability of AVTD
  Behavior within the Polarized $T^2$-symmetric vacuum spacetimes}}, 2021,
  Preprint. \href{http://arxiv.org/abs/2101.03167}{arXiv:2101.03167}.

\bibitem{andersson2001}
L.~Andersson and A.~D Rendall, \emph{Quiescent {{Cosmological Singularities}}},
  Commun. Math. Phys. \textbf{218} (2001), no.~3, 479--511,
  DOI:~\href{https://doi.org/10.1007/s002200100406}{10.1007/s002200100406}.

\bibitem{Andersson:2005}
L.~Andersson, H.~van Elst, W.~C. Lim, and C.~Uggla, \emph{{Asymptotic Silence
  of Generic Cosmological Singularities}}, Phys.\ Rev.\ Lett. \textbf{94}
  (2005), no.~5, 051101,
  DOI:~\href{https://doi.org/10.1103/PhysRevLett.94.051101}{10.1103/PhysRevLett.94.051101}.

\bibitem{belinskii1970}
V.~A. Belinskii, I.~M. Khalatnikov, and E.~M. Lifshitz, \emph{Oscillatory
  approach to a singular point in the relativistic cosmology}, Adv. Phys.
  \textbf{19} (1970), no.~80, 525--573,
  DOI:~\href{https://doi.org/10.1080/00018737000101171}{10.1080/00018737000101171}.

\bibitem{BERGER1997}
B.~K. Berger, P.~T. Chru\'{s}ciel, J.~Isenberg, and V.~Moncrief, \emph{Global
  foliations of vacuum spacetimes with {$T^2$} isometry}, Ann.\ Phys.
  \textbf{260} (1997), no.~1, 117--148,
  DOI:~\href{https://doi.org/10.1006/aphy.1997.5707}{10.1006/aphy.1997.5707}.

\bibitem{berger1998b}
B.~K. Berger and V.~Moncrief, \emph{Numerical evidence that the singularity in
  polarized {$U(1)$} symmetric cosmologies on {$T^3\times R$} is velocity
  dominated}, Phys. Rev. D \textbf{57} (1998), no.~12, 7235--7240,
  DOI:~\href{https://doi.org/10.1103/PhysRevD.57.7235}{10.1103/PhysRevD.57.7235}.

\bibitem{beyer2010b}
F.~Beyer and P.~G. LeFloch, \emph{Second-order hyperbolic {{Fuchsian}} systems
  and applications}, Class. Quantum Grav. \textbf{27} (2010), no.~24, 245012,
  DOI:~\href{https://doi.org/10.1088/0264-9381/27/24/245012}{10.1088/0264-9381/27/24/245012}.

\bibitem{beyer2017}
\bysame, \emph{Self\textendash{}gravitating fluid flows with {{Gowdy}} symmetry
  near cosmological singularities}, Commun. Partial. Differ. Equ. \textbf{42}
  (2017), no.~8, 1199--1248,
  DOI:~\href{https://doi.org/10.1080/03605302.2017.1345938}{10.1080/03605302.2017.1345938}.

\bibitem{beyer2020d}
F.~Beyer and T.~A. Oliynyk, \emph{Relativistic perfect fluids near {{Kasner}}
  singularities},  (2020), Preprint.
  \href{http://arxiv.org/abs/2012.03435}{arXiv:2012.03435}.

\bibitem{BOOS:2020}
F.~Beyer, T.~A. Oliynyk, and J.~A. Olvera-Santamaría, \emph{{The Fuchsian
  approach to global existence for hyperbolic equations}}, Commun. Partial.
  Differ. Equ. \textbf{46} (2021), 864--934.

\bibitem{choquet-bruhat2006}
Y.~{Choquet-Bruhat} and J.~Isenberg, \emph{Half polarized {$U(1)$}-symmetric
  vacuum spacetimes with {{AVTD}} behavior}, J. Geom. Phys. \textbf{56} (2006),
  no.~8, 1199--1214,
  DOI:~\href{https://doi.org/10.1016/j.geomphys.2005.06.011}{10.1016/j.geomphys.2005.06.011}.

\bibitem{choquet-bruhat2004}
Y.~{Choquet-Bruhat}, J.~Isenberg, and V.~Moncrief, \emph{Topologically general
  {$U(1)$} symmetric vacuum space-times with {{AVTD}} behavior}, Nuovo Cim. B
  \textbf{119} (2004), no.~7-9, 625--638,
  DOI:~\href{https://doi.org/10.1393/ncb/i2004-10174-x}{10.1393/ncb/i2004-10174-x}.

\bibitem{CIM1990}
P.~T. Chru{\'s}ciel, J.~Isenberg, and V.~Moncrief, \emph{{Strong cosmic
  censorship in polarised Gowdy spacetimes}}, Class. Quantum Grav. \textbf{7}
  (1990), no.~10, 1671--1680.

\bibitem{ChruscielKlinger:2015}
P.~T. Chr\'{u}sciel and P.~Klinger, \emph{Vacuum spacetimes with controlled
  singularities and without symmetries}, Phys. Rev. D \textbf{92} (2015),
  041501.

\bibitem{CHRUSCIEL:1990}
P.~T. Chruściel, \emph{On space-times with {$U(1)\times U(1)$} symmetric
  compact {Cauchy} surfaces}, Ann.\ Phys. \textbf{202} (1990), no.~1, 100 --
  150.

\bibitem{claudel1998a}
C.~M. Claudel and K.~P. Newman, \emph{The {{Cauchy Problem}} for
  {{Quasi}}-{{Linear Hyperbolic Evolution Problems}} with a {{Singularity}} in
  the {{Time}}}, Proc. R. Soc. Lond. A: Math. Phys. Eng. Sci. \textbf{454}
  (1998), no.~1972, 1073--1107.

\bibitem{Clausen2007}
A.~Clausen and J.~Isenberg, \emph{{Areal foliation and asymptotically
  velocity-term dominated behavior in T2 symmetric space-times with positive
  cosmological constant}}, J. Math. Phys. \textbf{48} (2007), no.~8, 082501.

\bibitem{xcoba}
{D. Yllanes and J. M. Mart\'{i}n-Garc\'{i}a}, \emph{{xCoba}: General component
  tensor computer algebra}, Available at \url{http://www.xact.es/xCoba},
  version 0.8.5.

\bibitem{damour2002}
T.~Damour, M.~Henneaux, A.~D. Rendall, and M.~Weaver, \emph{Kasner-{{Like
  Behaviour}} for {{Subcritical Einstein}}-{{Matter Systems}}}, Ann. Henri
  Poincar{\'e} \textbf{3} (2002), no.~6, 1049--1111,
  DOI:~\href{https://doi.org/10.1007/s000230200000}{10.1007/s000230200000}.

\bibitem{Eardley:1972}
D.~M. Eardley, E.~Liang, and R.~K. Sachs, \emph{{Velocity-Dominated
  Singularities in Irrotational Dust Cosmologies}}, J. Math. Phys. \textbf{13}
  (1972), no.~1, 99.

\bibitem{FOW:2021}
D.~Fajman, T.~A. Oliynyk, and Z.~Wyatt, \emph{Stabilizing relativistic fluids
  on spacetimes with non-accelerated expansion}, Commun. Math. Phys.
  \textbf{383} (2021), 401--426.

\bibitem{fournodavlos2021}
G.~Fournodavlos, \emph{Future dynamics of {{FLRW}} for the massless-scalar
  field system with positive cosmological constant},  (2021), Preprint.
  \href{http://arxiv.org/abs/2104.03771}{arXiv:2104.03771}.

\bibitem{Fournodavlos:2020}
G.~Fournodavlos and J.~Luk, \emph{{Asymptotically Kasner-like singularities}},
  arXiv.org (2020),
  Preprint.~\href{https://arxiv.org/abs/2003.13591v1}{arXiv:2003.13591v1}.

\bibitem{fournodavlos2020b}
G.~Fournodavlos, I.~Rodnianski, and J.~Speck, \emph{Stable {Big} {Bang}
  formation for {Einstein}'s equations: {The} complete sub-critical regime},
  (2020), Preprint.~\href{https://arxiv.org/abs/2012.05888}{arXiv:2012.05888}.

\bibitem{Garfinkle2021}
D.~Garfinkle, private communication.

\bibitem{Gowdy1974}
R.~H. Gowdy, \emph{{Vacuum spacetimes with two-parameter spacelike isometry
  groups and compact invariant hypersurfaces: Topologies and boundary
  conditions}}, Ann. Phys. \textbf{83} (1974), no.~1, 203--241.

\bibitem{hawkingLargeScaleStructure1973}
S.~W. Hawking and G.~F.~R. Ellis, \emph{The large scale structure of
  space-time}, first ed., {Cambridge University Press}, 1973,
  DOI:~\href{https://doi.org/10.1017/CBO9780511524646}{10.1017/CBO9780511524646}.

\bibitem{heinzle2012}
J.~M. Heinzle and P.~Sandin, \emph{The {{Initial Singularity}} of {{Ultrastiff
  Perfect Fluid Spacetimes Without Symmetries}}}, Commun. Math. Phys.
  \textbf{313} (2012), no.~2, 385--403,
  DOI:~\href{https://doi.org/10.1007/s00220-012-1496-x}{10.1007/s00220-012-1496-x}.

\bibitem{isenberg1999}
J.~Isenberg and S.~Kichenassamy, \emph{Asymptotic behavior in polarized
  {$T^2$}-symmetric vacuum space\textendash{}times}, J. Math. Phys. \textbf{40}
  (1999), no.~1, 340--352,
  DOI:~\href{https://doi.org/10.1063/1.532775}{10.1063/1.532775}.

\bibitem{Isenberg:1990}
J.~Isenberg and V.~Moncrief, \emph{{Asymptotic behavior of the gravitational
  field and the nature of singularities in Gowdy spacetimes}}, Ann.\ Phys.
  \textbf{199} (1990), no.~1, 84--122.

\bibitem{isenberg2002}
\bysame, \emph{Asymptotic behaviour in polarized and half-polarized {$U(1)$}
  symmetric vacuum spacetimes}, Class. Quantum Grav. \textbf{19} (2002),
  no.~21, 5361--5386,
  DOI:~\href{https://doi.org/10.1088/0264-9381/19/21/305}{10.1088/0264-9381/19/21/305}.

\bibitem{IsenbergWeaver:2003}
J.~{Isenberg} and M.~{Weaver}, \emph{{On the area of the symmetry orbits in
  T$^{2}$ symmetric spacetimes}}, Class. and Quantum Grav. \textbf{20} (2003),
  no.~16, 3783--3796.

\bibitem{xact}
{J. M. Mart\'{i}n-Garc\'{i}a}, \emph{{xAct}: Efficient tensor computer algebra
  for the {Wolfram} language}, Available at \url{http://www.xact.es}, version
  1.1.4.

\bibitem{kasner1921}
E.~Kasner, \emph{Geometrical {{Theorems}} on {{Einstein}}'s {{Cosmological
  Equations}}}, Am. J. Math. \textbf{43} (1921), no.~4, 217,
  DOI:~\href{https://doi.org/10.2307/2370192}{10.2307/2370192}.

\bibitem{kichenassamy2007k}
S.~Kichenassamy, \emph{Fuchsian {{Reduction}}}, Progress in {{Nonlinear
  Differential Equations}} and {{Their Applications}}, vol.~71, {Birkh\"auser
  Boston}, {Boston, MA}, 2007,
  DOI:~\href{https://doi.org/10.1007/978-0-8176-4637-0}{10.1007/978-0-8176-4637-0}.

\bibitem{kichenassamy1998}
S.~Kichenassamy and A.~D Rendall, \emph{Analytic description of singularities
  in {{Gowdy}} spacetimes}, Class. Quantum Grav. \textbf{15} (1998), no.~5,
  1339--1355,
  DOI:~\href{https://doi.org/10.1088/0264-9381/15/5/016}{10.1088/0264-9381/15/5/016}.

\bibitem{LeFlochWei:2021}
P.~G. LeFloch and C.~Wei, \emph{The nonlinear stability of self-gravitating
  irrotational {C}haplygin fluids in a {FLRW} geometry}, Ann. Henri Poincar\'e
  C \textbf{38} (2021), 757--814.

\bibitem{lifshitz1963}
E.~M. Lifshitz and I.~M. Khalatnikov, \emph{Investigations in relativistic
  cosmology}, Adv. Phys. \textbf{12} (1963), no.~46, 185--249,
  DOI:~\href{https://doi.org/10.1080/00018736300101283}{10.1080/00018736300101283}.

\bibitem{LiuOliynyk:2018b}
C.~Liu and T.~A. Oliynyk, \emph{Cosmological {N}ewtonian limits on large
  spacetime scales}, Commun. Math. Phys. \textbf{364} (2018), 1195--1304.

\bibitem{LiuOliynyk:2018a}
\bysame, \emph{Newtonian limits of isolated cosmological systems on long time
  scales}, Ann. \ Henri \ Poincar{\'e} \textbf{19} (2018), 2157--2243.

\bibitem{LiuWei:2021}
C.~Liu and C.~Wei, \emph{Future stability of the {FLRW} spacetime for a large
  class of perfect fluids}, Ann. Henri Poincar{\'e} \textbf{22} (2021),
  715--779.

\bibitem{Lott:2020b}
J.~Lott, \emph{{Kasner-like regions near crushing singularities}},  (2020),
  Preprint.~\href{https://arxiv.org/abs/2008.02674v2}{arXiv:2008.02674v2}.

\bibitem{Lott:2020a}
\bysame, \emph{{On the initial geometry of a vacuum cosmological spacetime}},
  Class. Quant. Grav. \textbf{37} (2020), no.~8, 085017.

\bibitem{Oliynyk:CMP_2016}
T.~A. Oliynyk, \emph{Future stability of the {FLRW} fluid solutions in the
  presence of a positive cosmological constant}, Commun. Math. Phys.
  \textbf{346} (2016), 293--312; see the preprint [arXiv:1505.00857] for a
  corrected version.

\bibitem{Oliynyk:2021}
\bysame, \emph{Future global stability for relativistic perfect fluids with
  linear equations of state $p={K}\rho$ where $1/3<{K}<1/2$}, SIAM J. Math.
  Anal. (accepted) (2021),
  Preprint.~\href{https://arxiv.org/abs/2002.12526}{arXiv:2002.12526}.

\bibitem{OliynykOlvera:2021}
T.~A. Oliynyk and J.~A. Olvera-Santamar\'{i}a, \emph{A {F}uchsian viewpoint on
  the weak-null condition}, J. Differential Eqns. \textbf{296} (2021),
  107--141.

\bibitem{rendall2000}
A.~D Rendall, \emph{Fuchsian analysis of singularities in {{Gowdy}} spacetimes
  beyond analyticity}, Class. Quantum Grav. \textbf{17} (2000), no.~16,
  3305--3316,
  DOI:~\href{https://doi.org/10.1088/0264-9381/17/16/313}{10.1088/0264-9381/17/16/313}.

\bibitem{ringstrom2009a}
H.~Ringstr{\"o}m, \emph{Strong cosmic censorship in ${T}^3$-{Gowdy}
  spacetimes}, Ann. Math. \textbf{170} (2009), no.~3, 1181--1240.

\bibitem{ringstrom2017}
H.~Ringstr\"{o}m, \emph{Linear systems of wave equations on cosmological
  backgrounds with convergent asymptotics}, Ast\'{e}risque (2020), no.~420,
  1--526. \MR{4163813}

\bibitem{Ringstrom:2021a}
H.~Ringstr{\"o}m, \emph{{On the geometry of silent and anisotropic big bang
  singularities}},  (2021),
  Preprint.~\href{https://arxiv.org/abs/2101.04955v1}{arXiv:2101.04955v1}.

\bibitem{Ringstrom:2021b}
\bysame, \emph{{Wave equations on silent big bang backgrounds}},  (2021),
  Preprint.~\href{https://arxiv.org/abs/2101.04939v1}{arXiv:2101.04939v1}.

\bibitem{Rodnianski2018HighD}
I~{Rodnianski} and J~{Speck}, \emph{{On the nature of Hawking's incompleteness
  for the Einstein-vacuum equations: The regime of moderately spatially
  anisotropic initial data}},  (2018),
  Preprint.~\href{https://arxiv.org/abs/1804.06825}{arXiv:1804.06825}.

\bibitem{rodnianski2014}
I.~Rodnianski and J.~Speck, \emph{Stable {{Big Bang}} formation in
  near-{{FLRW}} solutions to the {{Einstein}}-scalar field and
  {{Einstein}}-stiff fluid systems}, Sel. Math. New Ser. \textbf{24} (2018),
  no.~5, 4293--4459,
  DOI:~\href{https://doi.org/10.1007/s00029-018-0437-8}{10.1007/s00029-018-0437-8}.

\bibitem{Smulevici:2011}
J.~Smulevici, \emph{{On the area of the symmetry orbits of cosmological
  spacetimes with toroidal or hyperbolic symmetry}}, Anal. PDE \textbf{4}
  (2011), no.~2.

\bibitem{stahl2002}
F.~St{\aa}hl, \emph{Fuchsian analysis of {$S^2\times S^1$} and {$S^3$}
  {{Gowdy}} spacetimes}, Class. Quantum Grav. \textbf{19} (2002), no.~17,
  4483--4504,
  DOI:~\href{https://doi.org/10.1088/0264-9381/19/17/301}{10.1088/0264-9381/19/17/301}.

\bibitem{TaylorIII:1996}
M.~E. Taylor, \emph{Partial differential equations {III}: {N}onlinear
  equations}, Springer, 1996.

\bibitem{Weaver:2001}
M.~{Weaver}, B.~K. {Berger}, and J.~{Isenberg}, \emph{{Oscillatory Approach to
  the Singularity in Vacuum T$^{2}$ Symmetric Spacetimes}}, The Ninth Marcel
  Grossmann Meeting (Vahe~G. {Gurzadyan}, Robert~T. {Jantzen}, and Remo
  {Ruffini}, eds.), December 2002, pp.~1011--1012.

\end{thebibliography}

\end{document}